\documentclass[twoside]{article}

\newcommand{\dl}{\textit{dl}}
\newcommand{\Dom}{\ensuremath{\textit{Dom}\,}}
\newcommand{\BF}{\ensuremath{\textit{BF}_A}}
\newcommand{\PS}{\ensuremath{C_A}}
\newcommand{\fr}{\ensuremath{\mathit{se}}}

\newcommand{\export}{\mathbin{\setlength{\unitlength}{1ex}
     \begin{picture}(2.0,1.8)(-.8,0)
     \put(-.5,1.6){\line(1,0){1.4}}
     \put(-.5,-0.2){\line(1,0){1.4}}
     \put(-.44,-0.2){\line(0,1){1.8}}
     \put(.84,-0.2){\line(0,1){1.8}}
     \end{picture}
     }}

\newcommand{\SigSCL}{\ensuremath{\Sigma_{\textup{SCL}}(A)}}
\newcommand{\SigCP}{\ensuremath{\Sigma_{\textup{CP}}(A)}}
\newcommand{\Mse}{\ensuremath{\mathbb{M}_{\textit{se}}}}

\newcommand{\leftand}{~
     \mathbin{\setlength{\unitlength}{.9ex}
     \begin{picture}(1.6,1.8)(-.4,0)
     \put(-.8,0){\small$\wedge$}
     \put(-.68,-0.16){\textcolor{white}{\circle*{0.6}}}
     \put(-.68,-0.16){\circle{0.6}}
     \end{picture}
     }}
 
\newcommand{\leftands}{~
     \mathbin{\setlength{\unitlength}{.9ex}
     \begin{picture}(1.6,1.8)(-.4,0)
     \put(-.8,0){\small$\wedge$}
     \put(-.65,-0.1){\textcolor{white}{\circle*{0.6}}}
     \put(-.65,-0.1){\circle{0.6}}
     \end{picture}
     }}

\newcommand{\fulland}{~
     \mathbin{\setlength{\unitlength}{.9ex}
     \begin{picture}(1.6,1.8)(-.4,0)
     \put(-.8,0){\small$\wedge$}
     \put(-.68,-0.16){\circle*{0.66}}
     \end{picture}
     }}
          
\newcommand{\leftor}{~
     \mathbin{\setlength{\unitlength}{.9ex}
     \begin{picture}(1.6,1.8)(-.4,0)
     \put(-.8,0){\small$\vee$}
     \put(-.66,1.48){\textcolor{white}{\circle*{0.6}}}
     \put(-.66,1.48){\circle{0.6}}
     \end{picture}
     }}

\newcommand{\leftors}{~
     \mathbin{\setlength{\unitlength}{.9ex}
     \begin{picture}(1.6,1.8)(-.4,0)
     \put(-.8,0){\small$\vee$}
     \put(-.64,1.5){\textcolor{white}{\circle*{0.6}}}
     \put(-.64,1.55){\circle{0.6}}
     \end{picture}
     }}

\newcommand{\fullor}{~
     \mathbin{\setlength{\unitlength}{.9ex}
     \begin{picture}(1.6,1.8)(-.4,0)
     \put(-.8,0){\small$\vee$}
     \put(-.66,1.484){\circle*{0.66}}
     \end{picture}
     }}

\makeatletter
\let \@sverbatim \@verbatim
\def \@verbatim {\@sverbatim \verbatimplus}
{\catcode`'=13 \gdef \verbatimplus{\catcode`'=13 \chardef '=13 }} 
\makeatother

\usepackage{etoolbox}
\makeatletter
\patchcmd{\l@section}
  {\hfil}
  {\leaders\hbox{\normalfont$\m@th\mkern \@dotsep mu\hbox{.}\mkern \@dotsep mu$}\hfill}
  {}{}
\makeatother

\usepackage{tikz,stmaryrd,amssymb,amsmath,amsfonts,amsthm,url,color,relsize}

\usepackage{textcomp}

\newcommand{\sub}[2]{\ensuremath{[#1 \mapsto #2]}}
\newcommand{\ssub}[4]{\ensuremath{[#1 \mapsto #2, #3 \mapsto #4]}}

\newcommand{\M}{\ensuremath{\mathbb{M}}}
\newcommand{\ST}{\ensuremath{\SP}}
\newcommand{\SNF}{\ensuremath{\textit{SNF}\,}}
\newcommand{\SE}{se}
\newcommand{\nfs}{\ensuremath f}
\newcommand{\EqFSCL}{\SCLe}
\newcommand{\T}{\NT}
\newcommand{\Tone}{\ensuremath{{\mathcal{T}_{A,\triangle}}}}
\newcommand{\subT}[2]{\ensuremath{[#1 \mapsto #2]}}
\newcommand{\cd}{cd}
\newcommand{\dd}{dd}
\newcommand{\tsd}{tsd}
\newcommand{\invs}{\ensuremath g}
\newcommand{\FSCL}{\axname{FSCL}}
\newcommand{\FreeSCL}{\FSCL}

\newcommand{\axname}[1]{\textup{\ensuremath{\textrm{#1}}}}
\newcommand{\CP}{\axname{CP}}
\newcommand{\CPandneg}{\CP(\neg,\leftand,\leftor)}

\newcommand{\SCLe}{\axname{EqFSCL}}
\newcommand{\SCLi}{\axname{EqFSCL}^-}
\newcommand{\SCL}{\axname{SCL}}
\newcommand{\RPSCL}{\axname{RPSCL}}

\newcommand{\SP}{\ensuremath{{\mathcal{S}_A}}}
\newcommand{\NT}{\ensuremath{{\mathcal{T}_A}}}

\newcommand{\tr}{\ensuremath{{\sf T}}}
\newcommand{\fa}{\ensuremath{{\sf F}}}
\newcommand{\true}{\ensuremath{\textit{true}}}
\newcommand{\false}{\ensuremath{\textit{false}}}

\newcommand{\lef}{\ensuremath{\triangleleft}}
\newcommand{\rig}{\ensuremath{\triangleright}}
\renewcommand{\unlhd}{\ensuremath{\scalebox{0.78}{\raisebox{.1pt}[0pt][0pt]{$\;\trianglelefteq\;$}}}}
\renewcommand{\unrhd}{\ensuremath{\scalebox{0.78}{\raisebox{.1pt}[0pt][0pt]{$\;\trianglerighteq\;$}}}}
\newcommand{\tlef}{\unlhd}
\newcommand{\trig}{\unrhd}

\newtheorem{theorem}{Theorem}[subsection]
\newtheorem{lemma}[theorem]{Lemma}
\newtheorem{proposition}[theorem]{Proposition}

\newtheorem{definition}[theorem]{Definition} 
\theoremstyle{definition}

\newtheorem*{thm:sclcpl*}{Theorem~\ref{thm:sclcpl}}
\newtheorem*{thm:nfs*}{Theorem~\ref{thm:nfs}}

\begin{document}

\title{An independent axiomatisation for free short-circuit logic}

\author{Alban Ponse \qquad Daan J.C. Staudt\\[1mm]
 {\small Section Theory of Computer Science, Informatics Institute}\\
 {\small  Faculty of Science, University of Amsterdam}\\
 {\small \url{https://staff.science.uva.nl/a.ponse}\qquad\url{https://www.daanstaudt.nl}}
}
\date{}
\maketitle

\begin{abstract}
Short-circuit evaluation denotes the semantics of 
propositional connectives in which the second
argument is evaluated only if the first argument does not suffice 
to determine the value of the expression.
Free short-circuit logic is the equational logic in which compound
statements are evaluated from left to right, while atomic evaluations
are not memorised throughout the evaluation, 
i.e., evaluations of distinct occurrences of
an atom in a compound statement may yield different truth values. 
We provide a simple
semantics for free SCL and an independent axiomatisation.
Finally, we discuss evaluation strategies, some other SCLs, and side effects.
\\[2mm]
\end{abstract}

{\small\tableofcontents}
\section{Introduction}

Short-circuit(ed) evaluation denotes the semantics of binary propositional connectives in which the second
argument is evaluated only if the first argument does not suffice 
to determine the value of the expression. In the setting of computer science, connectives
that prescribe short-circuit evaluation tend to have specific names or notations, such as 
Dijkstra's \textbf{cand} (conditional and) and \textbf{cor} (see~\cite{Dij76,Gries81}), 
or the short-circuited connectives \texttt{\&\&} and \texttt{||} as used in programming 
languages such as C, Go, Java, and Perl.
Short-circuit evaluation in C is discussed in e.g.,
\cite{ZD03}, or in a context with partial predicates, in~\cite{McCarthy}.

A motivation for short-circuit evaluation arises in the setting in which 
the evaluation of atomic propositions can be state dependent and 
atomic evaluations may change the (evaluation) state.
As an example, consider the short-circuit evaluation of this program fragment: 
\[
\texttt{(f(x)$\:$>$\;$5)}\texttt{ \&\& } \texttt{(g(x)$\:$<$\;$3)}
\]
(which can occur in the condition of an if-then-else or while construct),
the result of which can be different from 
the short-circuit evaluation of 
\[
\texttt{(g(x)$\:$<$\;$3)}\texttt{ \&\& }\texttt{(f(x)$\:$>$\;$5)}\]
if a 
\emph{side effect} in the evaluation of the atomic propositions
\texttt{(f(x)$\:$>$\;$5)} or \texttt{(g(x)$\:$<$\;$3)}
changes the value of $\texttt x$.
Note that in this example, the expressions \texttt{(f(x)$\:$>$\;$5)}
and \texttt{(g(x)$\:$<$\;$3)} are considered to be propositional
variables, or, as we will henceforth call these, ``atoms''.
In Section~\ref{sec:4} we briefly discuss side effects and an example 
that shows that short-circuit conjunction is not a commutative operation. 

Following~\cite{BBR95} we write $P\leftand Q$ 
for the sequential conjunction of $P$ and $Q$ that prescribes short-circuit evaluation 
(the small circle indicates that the left argument must be evaluated first). 
Similarly, we write $P\leftor Q$ for the sequential disjunction of $P$ and $Q$ 
that prescribes short-circuit evaluation.

Another motivation for short-circuit evaluation arises in the setting
in which intermediate evaluation results are not at all memorised 
throughout the evaluation
of a propositional statement, i.e., evaluations of distinct occurrences of
an atom in a propositional statement may yield different truth values. 
A simple example of this phenomenon, taken from~\cite{BP10}, is the compound statement a 
pedestrian evaluates before crossing a road with two-way traffic driving on the right:
\[
\textit{look-left-and-check}\leftand(\textit{look-right-and-check}\leftand
\textit{look-left-and-check}).
\]
This statement requires one, or two, or three atomic evaluations and cannot be
simplified to one that requires less. In particular, the evaluation result of the
second occurrence of the atom \textit{look-left-and-check} may be \false, while its
first occurrence was evaluated \true. 
Observe that the associative variant 
\[(\textit{look-left-and-check}\leftand\textit{look-right-and-check})\leftand
\textit{look-left-and-check}\]
prescribes the same short-circuit evaluation.

In this paper we restrict evaluations to the truth values \true\ and \false\
(in the conclusions we include a few words on a truth value for ``undefined'').
Given this restriction, a natural question is 
``which logical laws axiomatise short-circuit evaluation?'', and in 
this paper we provide an answer by considering \emph{short-circuit logic} (SCL).
Different SCLs can be distinguished based on the extent to which atomic evaluation
results are memorised. We discuss in detail the SCL associated
with the last example, which is called \emph{free short-circuit logic} (\FSCL),
thus the short-circuit logic in which the second evaluation of an atom can be
different from its first evaluation.
With help of \emph{evaluation trees} we can give a simple and natural definition of 
short-circuit evaluation, 
and we provide a complete and independent equational axiomatisation of \FSCL.

The paper is structured as follows:
in Section~\ref{sec:2}, we define evaluation trees, equational axioms for \FSCL, and normal forms.
In Section~\ref{sec:3} we prove that these axioms are complete for the restriction to closed terms. 
In Section~\ref{sec:4} we recall \FSCL\ as defined earlier,
and consider evaluation strategies and
some other variants of SCL that identify more
propositional statements, and side effects.
We end the paper in Section~\ref{sec:Conc} with some conclusions.
The paper contains three appendices, containing detailed proofs on independence,
normalisation, and a quote of an earlier completeness proof.
\\[2mm]
\textit{Note.}
Considerable parts of the text in the forthcoming  sections stem from~\cite{Daan,BPS13}.

\section{Evaluation trees, axioms for free short-circuit logic, and normal forms}
\label{sec:2}
We define evaluation trees and 
provide an equational axiomatisation for free short-circuit logic
(Section~\ref{subsec:SCLe}).
Then we define  normal
forms for closed propositional statements (Section~\ref{subsec:snf}).

\subsection{Evaluation trees and axioms}
\label{subsec:SCLe}
Given a non-empty set $A$ of atoms, we define evaluation trees,
where \tr\ represents the truth value \true, and \fa\ represents 
the truth value \false.

\begin{definition}
\label{def:treesN}
The set \NT\ of \textbf{evaluation trees} over $A$ with leaves in 
$\{\tr, \fa\}$ is defined inductively by
\[\tr\in\NT,\quad\fa\in\NT, \quad (X\unlhd a\unrhd Y)\in\NT 
~\text{ for any }X,Y \in \NT \text{  and } a\in A.\]
The operator $\_\unlhd a\unrhd\_$ is called 
\textbf{tree composition over $a$}.
In the evaluation tree $X \unlhd a \unrhd Y$, 
the root is represented by $a$,
the left branch by $X$ and the right branch by $Y$. 
The \textbf{depth} $d(..)$ of an evaluation tree is defined by 
\[d(\tr) = d(\fa) = 0\quad\text{and}\quad
d(Y \unlhd a \unrhd Z) = 1 + \max(d(Y ), d(Z)).\] 
\end{definition}
We refer to trees in \NT\ as evaluation trees, or trees for short. 
Next to the formal notation for evaluation
trees we will also use a more pictorial representation. For example,
the tree
\[\fa\unlhd b\unrhd(\tr\unlhd a\unrhd\fa)\]
can be depicted as follows, where $\unlhd$ yields a left branch and $\unrhd$ a right branch:
\begin{equation}
\label{plaatje}
\hspace{-12mm}
\begin{tikzpicture}[%
      level distance=7.5mm,
      level 1/.style={sibling distance=15mm},
      level 2/.style={sibling distance=7.5mm},
      baseline=(current bounding box.center)]
      \node (a) {$b$}
        child {node (b1) {$\fa$}
        }
        child {node (b2) {$a$}
          child {node (d1) {$\tr$}} 
          child {node (d2) {$\fa$}}
        };
      \end{tikzpicture}
\end{equation}

In order to define a short-circuit semantics for negation and the sequential 
connectives, we first define the \emph{leaf replacement} operator, 
`replacement' for short, on trees in \NT\ as follows. 
For $X\in\NT$, the replacement of \tr\ with $Y$ and $\fa$ with $Z$ in $X$, denoted
\[X[\tr\mapsto Y, \fa \mapsto Z]\]
is defined recursively by 
\begin{align*}
\tr[\tr\mapsto Y,\fa\mapsto Z]&= Y,\\
\fa[\tr\mapsto Y,\fa\mapsto Z]&= Z, \\
(X_1\unlhd a\unrhd X_2)[\tr\mapsto Y,\fa\mapsto Z]
&=X_1[\tr\mapsto Y,\fa\mapsto Z]\unlhd a\unrhd X_2[\tr\mapsto Y,\fa\mapsto Z].
\end{align*}
We note that the order in which the replacements of leaves of 
$X$ is listed
is irrelevant and we adopt the convention of not listing  
identities inside the brackets, e.g., 
$X[\fa\mapsto Z]=X[\tr\mapsto \tr,\fa\mapsto Z]$.
Repeated replacements satisfy the following identity:
\begin{align}
\nonumber
X[\tr\mapsto Y_1,\fa\mapsto &Z_1][\tr\mapsto Y_2,\fa\mapsto Z_2]=\\
\label{id:rp}
&\quad
X[\tr\mapsto Y_1[\tr\mapsto Y_2,\fa\mapsto Z_2],~
\fa\mapsto Z_1[\tr\mapsto Y_2,\fa\mapsto Z_2]].
\end{align}
This identity easily follows by structural induction on evaluation trees. 

We define the
set \SP\ of closed (sequential) propositional statements over $A$
by the following grammar:
\[P ::= a\mid\tr\mid\fa\mid \neg P\mid P\leftand P\mid P\leftor P,\]
where $a\in A$, \tr\ is a constant for the truth  value \true, and \fa\ for \false,
and $\neg$ is negation. The underlying signature of \SP\ is
$\SigSCL=\{\leftand,\leftor,\neg,\tr,\fa,a\mid a\in A\}$.
We now have the terminology and notation to formally define the 
interpretation of propositional statements in \SP\ as evaluation trees
by a function $se$ (abbreviating  short-circuit evaluation).

\begin{definition}
\label{def:se}
The unary \textbf{short-circuit evaluation function} $se : \SP \to\NT$ 
is defined as
follows, where $a\in A$:
\begin{align*}
se(\tr) &= \tr,
&se(\neg P)&=se(P)[\tr\mapsto \fa,\fa\mapsto \tr],\\
se(\fa) &= \fa,
&se(P \leftand Q)&= se(P)[\tr\mapsto se(Q)],\\
se(a)&=\tr\unlhd a\unrhd \fa,
\quad&
se(P \leftor Q)&= se(P)[\fa\mapsto se(Q)].
\end{align*}
\end{definition}

The overloading of the notation \tr\ in $se(\tr)=\tr$ is harmless
and will turn out to be useful (and similarly for \fa).
As a simple example we derive the evaluation tree for $\neg b\leftand a$:
\[se(\neg b\leftand a)=se(\neg b)[\tr\mapsto se(a)]=
(\fa\unlhd b\unrhd\tr)[\tr\mapsto se(a)]=\fa\unlhd b\unrhd(\tr\unlhd a\unrhd\fa),\] 
which can be depicted as in~\eqref{plaatje}.
Also, $se(\neg(b\leftor\neg a))=\fa\unlhd b\unrhd(\tr\unlhd a\unrhd\fa)$.
An evaluation tree $se(P)$ represents short-circuit evaluation in a way that can be
compared to the notion of a truth table for propositional logic (PL) in that it 
represents each possible evaluation of $P$. However, there are some important differences with
truth tables: in $se(P)$, the sequentiality
of $P$'s evaluation is represented, and
the same atom may occur multiple times in $se(P)$.

We are interested in the set $\{P=Q\mid P,Q\in\SP \text{ and } se(P)=se(Q)\}$ and we
will show that all equations in this set are derivable from the 
axioms in Table~\ref{tab:SCL} by equational logic.
In~\cite{BP12a,BPS13} we defined \emph{free short-circuit logic}, notation \FSCL, 
and that is why this set of axioms 
is named \SCLe. In Section~\ref{subsec:4.2} we will return to the definition of \FSCL.

\begin{table}[t]
{ \small
\centering
\rule{1\textwidth}{.4pt}
\begin{align}
\tag{$E$}
\text{Axioms:}\phantom{Rules:}&
s=t
\quad\text{for all equations $s=t$ in $E$}\\
\label{el1}
\tag{Reflexivity}
&s=s
\quad\text{for every term $t$}\\
\label{el2}
\tag{Symmetry}
\text{Rules:}\phantom{Axioms:}&\dfrac{s=t}{t=s}
\\
\label{el3}
\tag{Transitivity}
&\dfrac{s=t, ~t=v}{s=v}
\\
\label{el4}
\tag{Congruence}
&\dfrac{s_1=t_1,~...,~s_n=t_n}{f(s_1,...,s_n)=f(t_1,...,t_n))}
\quad\text{for every n-ary $f$}
\\
\label{el5}
\tag{Substitution}
&\dfrac{s=t}{\sigma(s)=\sigma(t)} 
\quad\text{for $\sigma$ a substitution}
\end{align}
\hrule
}
\caption{Axioms and rules of the equational logic $E$} 
\label{tab:el}
\end{table}

The axioms and rules for equational logic for axiom set $E$ and terms $s,t$ are those
in Table~\ref{tab:el} (cf.~\cite{Burris}),
where we write $\sigma(s)$ for the application of substitution $\sigma$ 
to term $s$.
If for $\SigSCL$-terms $s$ and $t$, the equation $s=t$ 
is derivable from the axioms in \SCLe\ by equational
logic we write $\SCLe\vdash s=t$. 
Some comments on these axioms:

\begin{table}[t]
{ \small
\centering
\rule{1\textwidth}{.4pt}
\begin{align}
\label{SCL1}
\tag{F1}
\fa&=\neg\tr\\[0mm]
\label{SCL2}
\tag{F2}
x\leftors y&=\neg(\neg x\leftands\neg y)\\[0mm]
\label{SCL3}
\tag{F3}
\neg\neg x&=x\\[0mm]
\label{SCL4}
\tag{F4}
\tr\leftands x&=x\\[0mm]
\label{SCL5}
\tag{F5}
x\leftors\fa&=x\\[0mm]
\label{SCL6}
\tag{F6}
\fa\leftands x&=\fa\\[0mm]
\label{SCL7}
\tag{F7}
(x\leftands y)\leftands z&=x\leftands (y\leftands z)
\\[0mm]
\label{SCL8}
\tag{F8}
\qquad
\neg x\leftands \fa&= x \leftands\fa
\\[0mm]
\label{SCL9}
\tag{F9}
(x\leftands\fa)\leftors y
&=(x\leftors\tr)\leftands y\\
\label{SCL10}
\tag{F10}
(x\leftands y)\leftors(z\leftands\fa)&=
(x\leftors (z\leftands\fa))\leftands(y\leftors (z\leftands\fa))
\end{align}
\hrule
}
\caption{\SCLe, a set of axioms for \FSCL} 
\label{tab:SCL}
\end{table}

\begin{itemize}\setlength\itemsep{-1mm}
\item
Axioms~\eqref{SCL1} and~\eqref{SCL2} can be seen as defining equations for 
\fa\ and $\leftor$. 
\item
Axioms~\eqref{SCL1}-\eqref{SCL3} imply sequential versions of 
De~Morgan's laws, which implies a left-sequential version of
the duality principle. Below, we elaborate on this.
\item
Axioms~\eqref{SCL4}-\eqref{SCL7}
define some standard identities. 
\item
Axiom~\eqref{SCL8} illustrates 
a typical property of a logic that models immunity for side effects: 
although it is the case that for each $P\in\SP$,
the evaluation result of $P\leftand\fa$ is
\false, the evaluation of $P$ might also yield
a side effect. 
However, the same side effect and evaluation result
are obtained upon evaluation of $\neg P\leftand\fa$.
\item
Axiom~\eqref{SCL9} characterises another property that concerns
possible side effects: because the 
evaluation result of $P\leftand\fa$ for each possible evaluation  of 
the atoms in $P$ is \false, $Q$ is always evaluated in $(P\leftand\fa)\leftor Q$
and determines the evaluation result. For a similar reason, $Q$ is always evaluated 
in $(P\leftor\tr)\leftand Q$ and determines the evaluation result.
Note that the evaluations of $P\leftor\tr$ and $P\leftand\fa$ accumulate the same side effects,
which perhaps is more easily seen if one replaces $Q$ by either \tr\ or \fa.
\item
Axiom~\eqref{SCL10} defines a restricted form of 
right-distributivity of ${\leftor}$ over ${\leftand}$. 
This axiom holds because if $x$ evaluates to \true, both sides further evaluate
$y\leftor (z\leftand\fa)$, and if $x$ evaluates to \false,
$z\leftands\fa$ determines the further evaluation result (which is then \false,
and by axiom~\eqref{SCL6},  $y\leftor (z\leftand \fa)$ is not evaluated in the right-hand side).
\end{itemize} 

The dual of $P\in\SP$, notation $P^{\dl}$, is defined as follows (for $a\in A$):
\begin{align*}
\tr^{\dl}&=\fa,
&a^{\dl}&=a,
&(P\leftand Q)^{\dl}&= P^{\dl}\leftor Q^{\dl},\\
\fa^{\dl}&=\tr,
&(\neg P)^{\dl}&=\neg P^{\dl},
&(P\leftor Q)^{\dl}&= P^{\dl}\leftand Q^{\dl}.
\end{align*} 
The duality mapping $(\:)^{\dl}:\SP\to\SP$ is an involution, that is, $(P^{\dl})^{\dl}=P$.
Setting $x^{\dl}=x$ for each variable $x$, the duality principle
extends to equations, e.g.,
the dual of axiom~\eqref{SCL7} is $(x\leftor y)\leftor z = x\leftor (y\leftor z)$.
It immediately follows that \SCLe\ satisfies the duality principle, that is,
for all terms $s,t$ over \SigSCL,
\[\SCLe\vdash s=t\quad\iff\quad\SCLe\vdash s^{\dl}=t^{\dl}.\]

In order to use some standard notation and terminology of model theory,
we define a model that follows the definition of the function $se$ 
from Definition~\ref{def:se}.

\begin{definition}[The short-circuit evaluation model]
\label{def:Mse} 
Let \Mse\ be the $\SigSCL$-algebra with domain $\Dom(\Mse)=\{se(P)\mid P\in\SP\}$ in which
the interpretation of the constants is defined by
\begin{align*}
\llbracket \tr\rrbracket^{\Mse}&=\tr,\quad
\llbracket \fa\rrbracket^{\Mse}=\fa,\quad
\llbracket a\rrbracket^{\Mse}=\tr\unlhd a\unrhd \fa\quad\text{for all $a\in A$},
\end{align*}
and in which the interpretation of the connectives has equal notation and is defined by
\[\neg X=X[\tr\mapsto Y,\fa\mapsto\tr], 
\quad X\leftand Y=X[\tr\mapsto Y], 
\quad X\leftor Y=X[\fa\mapsto Y].\]
\end{definition}

We will show that $\Mse$ is an initial model for $\SCLe$, that is, 
if $X\in \Dom(\Mse)$ then for some $P\in\SP$, 
$X=\llbracket P\rrbracket^{\Mse}$,
and if $\Mse\models P=Q$, then $\SCLe \vdash P=Q$.
The first property holds by construction of $\Mse$.
Note that for all $P\in\SP,~\llbracket P\rrbracket^{\Mse}=se(P)$ 
(this follows easily by structural induction).
Hence, for all $P,Q\in\SP$,
$\Mse\models P=Q$ if, and only if, $se(P)=se(Q)$.
However, first of all, we have to show that \Mse\ is indeed a model of 
the equational logic \SCLe.

\begin{theorem}[Soundness]
\label{thm:sclsnd}
For all \SigSCL-terms $s,t$, if $\SCLe \vdash s=t$ then 
\(\Mse \models s=t.\)
\end{theorem}

\begin{proof}
It is immediately clear that reflexivity, symmetry, and transitivity hold.
For congruence we show only that 
for each \SigSCL-term $u$,
\[\Mse\models s=t \quad\Rightarrow\quad
\Mse\models u\leftand s=u\leftand t.
\]
Fix $u$ and let $i$ be an interpretation of variables in $\Mse$. 
We have to show $\Mse,i\models u\leftand s=u\leftand t$ if $\Mse\models s=t$.
Now there must be $P,Q,R\in\SP$ with $i(s)=se(P)$, $i(t)=se(Q)$, and $i(u)=se(R)$.
If $\Mse\models s=t$, then $\Mse,i\models s=t$, and thus $se(P) = se(Q)$. 
Hence,
\begin{equation*}
se(R)\subT{\tr}{se(P)}=
se(R)\subT{\tr}{se(Q)},
\end{equation*}
and thus $\Mse,i\models u\leftand s=u\leftand t$.
In a similar way it follows that substitution holds. 

Verifying the validity of the axioms in $\EqFSCL$ is cumbersome, but not
difficult. As an example we show this for \eqref{SCL3}. 
Fix some interpretation $i$ 
of variables and assume $i(x)=se(P)$, then 
\begin{equation*}
\Mse,i\models\neg\neg\SE(P) = \SE(P)\ssub{\tr}{\fa}{\fa}{\tr}
\ssub{\tr}{\fa}{\fa}{\tr}  = \SE(P),
\end{equation*}
where the latter equality follows from identity~\eqref{id:rp} for repeated 
replacements.
\end{proof}

The following result is non-trivial and proved in Section~\ref{subsec:cpl}.
\begin{theorem}[Completeness of \SCLe\ for closed terms]
\label{thm:sclcplP}
For all $P,Q\in\SP$,
\[\SCLe\vdash P=Q 
\quad\iff\quad
 \Mse\models P=Q.\]
\end{theorem}

Observe that \SCLe\ does \emph{not} imply the following properties:
\begin{itemize}\setlength\itemsep{-1mm}
\item idempotence, e.g., $se(a\leftand a)\ne se(a)$,
\item commutativity, e.g., $se(a\leftand b)\ne se(b\leftand a)$,
\item absorption, e.g., $se(a\leftand (a\leftor b))\ne se(a)$,
\item distributivity, 
e.g., $se((a\leftand b)\leftor c)\ne se((a\leftor c)\leftand (b\leftor c))$.
\end{itemize}

The following lemma is used in the proof of Theorem~\ref{thm:sclcplP}.
We note that the lemma's identity
was presented as an \SCLe-axiom in~\cite{BP12a}
and is now replaced by the current axiom~\eqref{SCL8}.
\begin{lemma}
\label{lem:seqs}
$\EqFSCL\vdash (x \leftor y) \leftand (z \leftand \fa) = 
(\neg x \leftor (z
  \leftand \fa)) \leftand (y \leftand (z \leftand \fa))$.
\end{lemma}
\begin{proof}
\begin{align*}
(x\leftor y) \leftand (z \leftand \fa)
&= (x \leftor y) \leftand ((z \leftand \fa) \leftand \fa)
&&\text{by \eqref{SCL6}, \eqref{SCL7}} \\
&= ((x \leftor y) \leftand \neg(z \leftand \fa)) \leftand \fa
&&\text{by \eqref{SCL8}, \eqref{SCL7}} \\
&= ((\neg x \leftand \neg y) \leftor (z \leftand \fa)) \leftand \fa
&&\text{by \eqref{SCL8}, \eqref{SCL2}, \eqref{SCL3}} \\
&= ((\neg x \leftor (z \leftand \fa)) \leftand (\neg y \leftor (z
  \leftand \fa))) \leftand \fa
&&\text{by \eqref{SCL10}} \\
&= (\neg x \leftor (z \leftand \fa)) \leftand ((\neg y \leftor (z
  \leftand \fa)) \leftand \fa)
&&\text{by \eqref{SCL7}} \\
&= (\neg x \leftor (z \leftand \fa)) \leftand ((y \leftand \neg(z
  \leftand \fa)) \leftand \fa)
&&\text{by \eqref{SCL8}, \eqref{SCL2}, \eqref{SCL3}} \\
&= ((\neg x \leftor (z \leftand \fa)) \leftand y) \leftand (\neg(z
  \leftand \fa) \leftand \fa)
&&\text{by \eqref{SCL7}} \\
&= ((\neg x \leftor (z \leftand \fa)) \leftand y) \leftand ((z
  \leftand \fa) \leftand \fa)
&&\text{by \eqref{SCL8}} \\
&= ((\neg x \leftor (z \leftand \fa)) \leftand y) \leftand (z
  \leftand \fa)
&&\text{by \eqref{SCL7}, \eqref{SCL6}} \\
&= (\neg x \leftor (z \leftand \fa)) \leftand (y \leftand (z
  \leftand \fa)).
&&\text{by \eqref{SCL7}} 
\end{align*}
\end{proof}

We conclude this section with some facts about $\SCLe$.
First, we prove that axioms~\eqref{SCL1} and~\eqref{SCL3} are derivable from
the remaining axioms, and then we show that these remaining axioms are independent.
For both results, we used tools from~\cite{BirdBrain}.
We note that Lemma~\ref{lem:seqs} was also checked with the theorem
prover \emph{Prover9} from~\cite{BirdBrain}.

\begin{definition}
\label{def:SCLe'}
Let $\SCLi=\SCLe\setminus\{\eqref{SCL1},\eqref{SCL3}\}$.
\end{definition}

\begin{proposition}
\label{prop:extraF}
$\SCLi\setminus\{\eqref{SCL8},\eqref{SCL10}\}
\vdash\eqref{SCL1},\eqref{SCL3}$.
\end{proposition}

\begin{proof}
Distilled from output of \emph{Prover9}~\cite{BirdBrain}.
In order to derive axiom~\eqref{SCL1} we start with some auxiliary results:
\begin{align}
\nonumber
\neg x\leftand\neg\fa
&=(\neg x\leftand\neg\fa)\leftor\fa
&&\text{by~\eqref{SCL5}}\\
\nonumber
&=\neg(\neg(\neg x\leftand\neg\fa)\leftand\neg\fa)
&&\text{by~\eqref{SCL2}}\\
\nonumber
&=\neg((x\leftor\fa)\leftand\neg\fa)
&&\text{by~\eqref{SCL2}}\\
\label{Aux1}
\tag{Aux1}
&=\neg(x\leftand\neg\fa),
&&\text{by~\eqref{SCL5}}
\end{align}
hence, 
\begin{align}
\label{Aux2}
\tag{Aux2}
&\neg\neg x\leftand\neg\fa
\stackrel{\eqref{Aux1}}=\neg(\neg x\leftand\neg\fa)
\stackrel{\eqref{SCL2}}=x\leftor\fa
\stackrel{\eqref{SCL5}}=x,
\\
\label{Aux3}
\tag{Aux3}
&\neg\neg\fa
\stackrel{\eqref{SCL4}}=\neg(\tr\leftand\neg\fa)
\stackrel{\eqref{Aux1}}=\neg\tr\leftand\neg\fa,
\\
\label{Aux4}
\tag{Aux4}
&\neg \fa
\stackrel{\eqref{SCL6}}= \neg(\fa\leftand\neg\fa)
\stackrel{\eqref{Aux1}}=\neg\fa\leftand\neg\fa.
\end{align}
Next, 
\begin{align}
\nonumber
\fa
&= \neg\neg \fa\leftand\neg\fa
&&\text{by \eqref{Aux2}}\\
\nonumber
&= (\neg\tr\leftand\neg\fa)\leftand\neg\fa
&&\text{by \eqref{Aux3}}\\
\nonumber
&=\neg\tr\leftand(\neg\fa\leftand\neg\fa)
&&\text{by \eqref{SCL7}}\\
\label{Aux5}
\tag{Aux5}
&=\neg\tr\leftand\neg\fa,
&&\text{by \eqref{Aux4}}
\end{align}
hence,
\begin{align}
\label{Aux6}
\tag{Aux6}
&\neg\fa
\stackrel{\eqref{Aux5}}=
\neg(\neg\tr\leftand\neg\fa)
\stackrel{\eqref{SCL2}}=
\tr\leftor\fa\stackrel{\eqref{SCL5}}=\tr.
\end{align}
With these auxiliary results we derive axiom~\eqref{SCL1}:
\[\fa
\stackrel{\eqref{Aux5}}=
\neg\tr\leftand\neg\fa
\stackrel{\eqref{Aux1}}=
\neg(\tr\leftand\neg\fa)
\stackrel{\eqref{SCL4}}=
\neg\neg\fa\stackrel{\eqref{Aux6}}=\neg\tr.
\]

Finally, we derive axiom~\eqref{SCL3} and start with an auxiliary result:
\begin{equation}
\label{Aux7}
\tag{Aux7}
\fa \leftor \tr
\stackrel{\eqref{SCL2}}= \neg(\neg \fa\leftand\neg\tr)
\stackrel{\eqref{Aux6}}= \neg(\tr\leftand\neg\tr)
\stackrel{\eqref{SCL4}}=\neg\neg\tr
\stackrel{\eqref{SCL1}}=\neg\fa
\stackrel{\eqref{Aux6}}=\tr,
\end{equation}
and thus
\begin{align*}
\neg\neg x
&=\neg(\tr\leftand\neg x)
&&\text{by \eqref{SCL4}}\\
&=\neg(\neg\fa\leftand\neg x)
&&\text{by \eqref{Aux6}}\\
&=\fa\leftor x
&&\text{by~\eqref{SCL2}}\\
&=(\fa\leftand\fa)\leftor x
&&\text{by~\eqref{SCL6}}\\
&=(\fa\leftor\tr)\leftand x
&&\text{by~\eqref{SCL9}}\\
&=\tr\leftand x
&&\text{by~\eqref{Aux7}}\\
&=x.
&&\text{by~\eqref{SCL4}}
\end{align*}
\end{proof}

\begin{theorem}
\label{thm:indepSCLe'}
The axioms of $\SCLi$ are independent if $A$ contains at least two atoms.
\end{theorem}

\begin{proof} 
With the tool \emph{Mace4}~\cite{BirdBrain},
one easily obtains for each of the axioms of $\SCLi$ an independence model 
(a model in which that axiom is not valid, while all remaining axioms are). 
We show one of the eight cases here and defer the remaining cases to Appendix~\ref{App:scl}.

In order to prove  independence of axiom~\eqref{SCL10}, that is, 
\((x\leftand y)\leftor(z\leftand\fa)=(x\leftor (z\leftand\fa))\leftand(y\leftor (z\leftand\fa)),\)
assume $A\supseteq\{a,b\}$, and consider the model \M\ 
with domain $\{0,1,2,3\}$ in which the interpretation of the constants is defined by
\[\llbracket \tr\rrbracket^{\M}=1,\quad\llbracket \fa\rrbracket^{\M}=0,\quad\llbracket a\rrbracket^{\M}=2,\quad
\text{$\llbracket c\rrbracket^{\M}=3$ for all $c\in A\setminus\{a\}$},
\]
and in which the connectives are defined by
\[
\begin{array}{r@{\hspace{6pt}}|@{\hspace{6pt}}c}
\neg\\\hline\\[-4mm]
0&1\\
1&0\\
2&2\\
3&3
\end{array}
\qquad\qquad
\begin{array}{r@{\hspace{6pt}}|@{\hspace{6pt}}c@{\hspace{6pt}}c@{\hspace{6pt}}c@{\hspace{6pt}}c@{\hspace{6pt}}c}
\leftand&
0&1&2&3\\\hline\\[-4mm]
0&
0&0&0&0\\
1&
0&1&2&3\\
2&
0&2&0&0\\
3&
3&3&3&3
\end{array}
\qquad\qquad
\begin{array}{r@{\hspace{2mm}}|@{\hspace{2mm}}c@{\hspace{2mm}}c@{\hspace{2mm}}c@{\hspace{2mm}}c@{\hspace{2mm}}c}
\leftor&
0&1&2&3\\\hline\\[-4mm]
0&
0&1&2&3\\
1&
1&1&1&1\\
2&
2&1&1&1\\
3&
3&3&3&3
\end{array}
\] 
Then all axioms from 
$\SCLi\setminus\{\eqref{SCL10}\}$ are valid in \M,
while $\llbracket (a \leftand a) \leftor (b \leftand\fa)\rrbracket^{\M}=3$ and
$\llbracket (a \leftor (b \leftand\fa)) \leftand (a \leftor (b \leftand\fa))\rrbracket^{\M}=1$. 
\end{proof}

\subsection{Normal forms}
\label{subsec:snf}
To aid in the forthcoming proof of Theorem~\ref{thm:sclcplP}
we define normal forms for $\ST$-terms.
When considering trees in $\SE[\ST]$ (the image of $\SE$ for $\ST$-terms),
we note that some trees only have
$\tr$-leaves, some only $\fa$-leaves and some both $\tr$-leaves and
$\fa$-leaves. For any $\ST$-term $P$, \[\SE(P \leftor \tr)\] is a tree
with only $\tr$-leaves, as can easily be seen from the definition of $\SE$:
\[\SE(P \leftor \tr)=se(P)[\fa\mapsto \tr].\]
Similarly, for any $\ST$-term $P$, $\SE(P \leftand
\fa)$ is a tree with only $\fa$-leaves. 
The simplest trees in the
image of $\SE$ that have both types of leaves are $\SE(a)$ and $se(\neg a)$ for $a \in A$. 

We define the grammar for our normal form before we motivate it.

\begin{definition}
\label{def:snf}
A term $P \in \ST$ is said to be in \textbf{$\SCL$ Normal Form $(\SNF)$} if it
is generated by the following grammar:
\begin{align*}
P &::= P^\tr ~\mid~ P^\fa ~\mid~ P^\tr \leftand P^* 
&&(\SNF\text{-terms})\\
P^\tr &::= \tr ~\mid~ (a \leftand P^\tr) \leftor P^\tr &&(\tr\text{-terms})\\
P^\fa &::= \fa ~\mid~ (a \leftor P^\fa) \leftand P^\fa &&(\fa\text{-terms})\\[2mm]
P^* &::= P^c ~\mid~ P^d &&(*\text{-terms})\\
P^c &::= P^\ell ~\mid~ P^* \leftand P^d\\
P^d &::= P^\ell ~\mid~ P^* \leftor P^c
\\[2mm]
P^\ell &::= (a \leftand P^\tr ) \leftor P^\fa
  ~\mid~ (\neg a \leftand P^\tr ) \leftor P^\fa&&(\ell\text{-terms})
\end{align*}
where $a \in A$. We refer to $P^\tr$-forms as $\tr$-terms, to $P^\fa$-forms as
$\fa$-terms,
to $P^\ell$-forms as
$\ell$-terms (the name refers to literal terms), and to $P^*$-forms as $*$-terms.
Finally, a term of the form $P^\tr \leftand P^*$ is referred to as a
$\tr$-$*$-term.
\end{definition}

For each $\tr$-term $P$, $\SE(P)$ is a tree with only $\tr$-leaves.  
$\ST$-terms that have in their $se$-image only $\tr$-leaves will be rewritten to  
$\tr$-terms. Similarly, terms that have in their $se$-image only $\fa$-leaves 
will be rewritten to $\fa$-terms. 
Note that $\leftor$ is right-associative in $\tr$-terms, e.g.,
\[(a \leftand \tr) \leftor ((b \leftand \tr) \leftor \tr)
\quad\text{is a \tr-term, but 
$((a \leftand \tr) \leftor (b \leftand \tr)) \leftor \tr$ is not,}
\]
and that $\leftand$ is right-associative in $\fa$-terms.
Furthermore, the $se$-images of $\tr$-terms and $\fa$-terms follow a 
simple pattern: observe that for $P,Q\in P^\tr$,
$se((a\leftand P)\leftor Q)$
is of the form
\[
\begin{tikzpicture}[%
      level distance=7.5mm,
      level 1/.style={sibling distance=15mm},
      level 2/.style={sibling distance=7.5mm},
      baseline=(current bounding box.center)]
      \node (a) {$a$}
        child {node (b1) {$se(P)$}
        }
        child {node (b2) {$se(Q)$}
        };
      \end{tikzpicture}
\]    

Before we discuss the $\tr$-$*$-terms | the third type of our $\SNF$ normal 
forms | we consider the $*$-terms, which are 
$\leftand$-$\leftor$-combinations of $\ell$-terms with the restriction 
that $\leftand$ and $\leftor$
associate to 
the left. This restriction is defined with help of the syntactical categories 
$P^c$ and $P^d$. 
From now on we shall use $P^\tr$, $P^*$, etc.~both to denote grammatical 
categories and as
variables for terms in those categories. As an example, 
\[(P^\ell\leftand Q^\ell)\leftand R^\ell\]
is  a $*$-term (it is in $P^c$-form), while
$P^\ell\leftand(Q^\ell\leftand R^\ell)$ is not a $*$-term.
We consider $\ell$-terms to be ``basic''
in $*$-terms in the sense that they are the smallest grammatical unit that 
generate
$se$-images in which both \tr\ and \fa\ occur. More precisely, the $se$-image
of an $\ell$-term
has precisely one node (its root) that has paths to both \tr\ and \fa. 

$\ST$-terms that have both \tr\ and \fa\ in their $se$-image
will be rewritten to
$\tr$-$*$-terms. A $\tr$-$*$-term is the conjunction of a
$\tr$-term 
and a $*$-term. The first conjunct is necessary to encode a term such as 
\[[a\leftor(b\leftor\tr)]\leftand c\]
where the evaluation values of $a$ and $b$ are not relevant, but where 
their side effects may influence the evaluation value of $c$, as can be clearly
seen from its $se$-image that has three different nodes 
that model the evaluation of $c$:
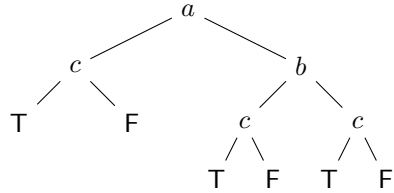
\begin{center}
\begin{tikzpicture}[%
level distance=7.5mm,
level 1/.style={sibling distance=30mm},
level 2/.style={sibling distance=15mm},
level 3/.style={sibling distance=7.5mm}
]
\node (a) {$a$}
  child {node (b1) {$c$}
    child {node (c1) {$\tr$}
    }
    child {node (c2) {$\fa$}
    }
  }
  child {node (b2) {$b$}
    child {node (c3) {$c$}
      child {node (d5) {$\tr$}} 
      child {node (d6) {$\fa$}}
    }
    child {node (c4) {$c$}
      child {node (d7) {$\tr$}} 
      child {node (d8) {$\fa$}}
    }
  };
\end{tikzpicture}
\end{center}
From this example it can be easily seen that the above 
\tr-$*$-term can be also represented
as the disjunction of an $\fa$-term and a $*$-term, namely of the \fa-term that
encodes $a\leftand(b\leftand\fa)$ and the $*$-term that encodes $c$, thus as
\[[(a\leftor\fa)\leftand((b\leftor\fa)\leftand\fa)]\leftor [(c\leftand\tr)\leftor\fa].\]
However, we chose to use a \tr-term and a conjunction for this purpose.

The remainder of this section is
concerned with defining and proving correct a normalisation function 
\[\nfs:\ST \to \SNF. 
\]
We will define $\nfs$ recursively using the functions
\begin{equation*}
\nfs^n: \SNF \to \SNF \quad\text{and}\quad
\nfs^c: \SNF \times \SNF \to \SNF.
\end{equation*}
The first of these will be used to rewrite negated $\SNF$-terms to $\SNF$-terms
and the second to rewrite the conjunction of two $\SNF$-terms to an
$\SNF$-term. By \eqref{SCL2} we have no need for a dedicated function that
rewrites the disjunction of two $\SNF$-terms to an $\SNF$-term.
The normalisation function $\nfs: \ST \to \SNF$ is defined
recursively, using $\nfs^n$ and $\nfs^c$, as follows.
\begin{align}
\nfs(a) &= \tr \leftand ((a \leftand \tr) \leftor \fa)
  \label{eq:nfs1} \\
\nfs(\tr) &= \tr
  \label{eq:nfs2} \\
\nfs(\fa) &= \fa
  \label{eq:nfs3} \\
\nfs(\neg P) &= \nfs^n(\nfs(P))
  \label{eq:nfs4} \\
\nfs(P \leftand Q) &= \nfs^c(\nfs(P), \nfs(Q))
  \label{eq:nfs5} \\
\nfs(P \leftor Q) &= \nfs^n(\nfs^c(\nfs^n(\nfs(P)), \nfs^n(\nfs(Q)))).
  \label{eq:nfs6}
\end{align}
Observe that $\nfs(a)$ is indeed the unique \tr-$*$-term 
with the property that $se(a)=se(\nfs(a))$, and also that
$se(\tr)=se(\nfs(\tr))$ and $se(\fa)=se(\nfs(\fa))$
(cf.~Theorem~\ref{thm:nfs}).
 
We proceed by defining $\nfs^n$. Analysing the semantics of $\tr$-terms and
$\fa$-terms together with the definition of $\SE$ on negations, it becomes
clear that $\nfs^n$ must turn $\tr$-terms into $\fa$-terms and vice versa.
We also remark that $\nfs^n$ must preserve the left-associativity of the
$*$-terms in $\tr$-$*$-terms, modulo the associativity within $\ell$-terms.
We define $\nfs^n: \SNF \to \SNF$ as follows, using the auxiliary function
$\nfs^n_1: P^* \to P^*$ to push in the negation symbols when
negating a $\tr$-$*$-term. We note that there is no ambiguity between the
different grammatical categories present in an $\SNF$-term, i.e., any
$\SNF$-term is in exactly one of the grammatical categories identified in
Definition~\ref{def:snf}, and that all right-hand sides are of the intended 
grammatical category.

\begin{align}
\nfs^n(\tr) &= \fa
  \label{eq:nfsn1} \\
\nfs^n((a \leftand P^\tr) \leftor Q^\tr) &= (a \leftor
  \nfs^n(Q^\tr)) \leftand \nfs^n(P^\tr)
  \label{eq:nfsn2} \\[2mm]
\nfs^n(\fa) &= \tr
  \label{eq:nfsn3} \\
\nfs^n((a \leftor P^\fa) \leftand Q^\fa) &= (a \leftand
  \nfs^n(Q^\fa)) \leftor \nfs^n(P^\fa)
  \label{eq:nfsn4}\\[2mm]
\nfs^n(P^\tr \leftand Q^*) &= P^\tr \leftand \nfs^n_1(Q^*)
  \label{eq:nfsn5} \\[2mm]
\nfs^n_1((a \leftand P^\tr) \leftor Q^\fa) &= (\neg a \leftand
  \nfs^n(Q^\fa)) \leftor \nfs^n(P^\tr)
  \label{eq:nfsn6} \\
\nfs^n_1((\neg a \leftand P^\tr) \leftor Q^\fa) &= (a \leftand
  \nfs^n(Q^\fa)) \leftor \nfs^n(P^\tr)
  \label{eq:nfsn7} \\
\nfs^n_1(P^* \leftand Q^d) &= \nfs^n_1(P^*) \leftor \nfs^n_1(Q^d)
  \label{eq:nfsn8} \\
\nfs^n_1(P^* \leftor Q^c) &= \nfs^n_1(P^*) \leftand \nfs^n_1(Q^c).
  \label{eq:nfsn9}
\end{align}

Now we turn to defining $\nfs^c$. We distinguish the following cases:
\begin{enumerate}\setlength\itemsep{-1mm}
\item[$(1)$] $\nfs^c(P^\tr, Q)$
\item[$(2)$] $\nfs^c(P^\fa, Q)$
\item[$(3)$] $\nfs^c(P^\tr\leftand P^*, Q)$
\end{enumerate}
In case $(1)$, it is apparent that the conjunction of a $\tr$-term with
another term always yields a term of the same grammatical category as the
second conjunct. We define $\nfs^c$ recursively by a 
case distinction on its first argument, and in the second case by a further 
case distinction on its second argument.  
\begin{align}
\nfs^c(\tr, P) &= P
  \label{eq:nfsc1} \\
\nfs^c((a \leftand P^\tr) \leftor Q^\tr, R^\tr) &= (a \leftand
  \nfs^c(P^\tr, R^\tr)) \leftor \nfs^c(Q^\tr, R^\tr)
  \label{eq:nfsc2} \\
\nfs^c((a \leftand P^\tr) \leftor Q^\tr, R^\fa) &= (a \leftor
  \nfs^c(Q^\tr, R^\fa)) \leftand \nfs^c(P^\tr, R^\fa)
  \label{eq:nfsc3} \\
\nfs^c((a \leftand P^\tr) \leftor Q^\tr, R^\tr \leftand S^*) &=
  \nfs^c((a \leftand P^\tr) \leftor Q^\tr, R^\tr) \leftand S^*.
  \label{eq:nfsc4}
\end{align}

For case $(2)$ (the first argument is an $\fa$-term) we make use
of \eqref{SCL6}. This immediately implies that the conjunction of an
$\fa$-term with another term is itself an $\fa$-term.
\begin{align}
\nfs^c(P^\fa, Q) &= P^\fa
  \label{eq:nfsc5}
\end{align}

For the remaining case $(3)$ (the first argument is an \tr-$*$-term)
we distinguish three sub-cases:
\begin{enumerate}\setlength\itemsep{-1mm}
\item[$(3.1)$] The second argument is a $\tr$-term,
\item[$(3.2)$] The second argument is an $\fa$-term, and
\item[$(3.3$)] The second argument is a \tr-$*$-term.
\end{enumerate}
For case $(3.1)$ we will use an auxiliary function
$\nfs^c_1: P^* \times P^\tr \to P^*$ to turn conjunctions of a $*$-term with
a $\tr$-term into $*$-terms. We define $\nfs^c_1$ recursively by a 
case distinction on its first argument. 
Together with \eqref{SCL7} (associativity) this allows us to
define $\nfs^c$ for this case. Observe that the right-hand
sides of the clauses defining $\nfs^c_1$ are indeed $*$-terms. 
\begin{align}
\nfs^c(P^\tr \leftand Q^*, R^\tr) &= P^\tr \leftand
  \nfs^c_1(Q^*, R^\tr) 
  \label{eq:nfsc6} \\[2mm]
\nfs^c_1((a \leftand P^\tr) \leftor Q^\fa, R^\tr) &= (a \leftand
  \nfs^c(P^\tr, R^\tr)) \leftor Q^\fa
  \label{eq:nfsc7} \\
\nfs^c_1((\neg a \leftand P^\tr) \leftor Q^\fa, R^\tr) &= (\neg a
  \leftand \nfs^c(P^\tr, R^\tr)) \leftor Q^\fa
  \label{eq:nfsc8} \\
\nfs^c_1(P^* \leftand Q^d, R^\tr) &= P^* \leftand \nfs^c_1(Q^d, R^\tr)
  \label{eq:nfsc9} \\
\nfs^c_1(P^* \leftor Q^c, R^\tr) &= \nfs^c_1(P^*, R^\tr) \leftor
  \nfs^c_1(Q^c, R^\tr).
  \label{eq:nfsc10}
\end{align}
For case $(3.2)$ we need to define 
$\nfs^c(P^\tr \leftand Q^*, R^\fa)$, which will 
be an
$\fa$-term. Using \eqref{SCL7} we reduce this problem to converting
$Q^*$ to an $\fa$-term, for which we use the auxiliary function
$\nfs^c_2: P^* \times P^\fa \to P^\fa$ that we define recursively by a 
case distinction on its first argument. Observe that the right-hand
sides of the clauses defining $\nfs^c_2$ are all $\fa$-terms. 
\begin{align}
\nfs^c(P^\tr \leftand Q^*, R^\fa) &= \nfs^c(P^\tr, \nfs^c_2(Q^*,
  R^\fa))
  \label{eq:nfsc11} \\[2mm]
\nfs^c_2((a \leftand P^\tr) \leftor Q^\fa, R^\fa) &= (a \leftor
  Q^\fa) \leftand \nfs^c(P^\tr, R^\fa)
  \label{eq:nfsc12} \\
\nfs^c_2((\neg a \leftand P^\tr) \leftor Q^\fa, R^\fa) &= (a
  \leftor \nfs^c(P^\tr, R^\fa)) \leftand Q^\fa
  \label{eq:nfsc13} \\
\nfs^c_2(P^* \leftand Q^d, R^\fa) &= \nfs^c_2(P^*, \nfs^c_2(Q^d,
  R^\fa))
  \label{eq:nfsc14} \\
\nfs^c_2(P^* \leftor Q^c, R^\fa) &= \nfs^c_2(\nfs^n(\nfs^c_1(P^*,
  \nfs^n(R^\fa))), \nfs^c_2(Q^c, R^\fa)).
  \label{eq:nfsc15}
\end{align}
For case $(3.3)$ we need to define $\nfs^c(P^\tr \leftand Q^*, R^\tr \leftand S^*)$.
We use the auxiliary function $\nfs^c_3: P^*
\times (P^\tr \leftand P^*) \to P^*$ to ensure that the result is a
$\tr$-$*$-term, and we define $\nfs^c_3$ by a case distinction on its second argument. 
Observe that the right-hand
sides of the clauses defining $\nfs^c_3$ are all $*$-terms. 
\begin{align}
\nfs^c(P^\tr \leftand Q^*, R^\tr \leftand S^*) &= P^\tr \leftand 
  \nfs^c_3(Q^*, R^\tr \leftand S^*)
  \label{eq:nfsc16} \\[2mm]
\nfs^c_3(P^*, Q^\tr \leftand R^\ell) &= \nfs^c_1(P^*, Q^\tr) \leftand
  R^\ell
  \label{eq:nfsc17} \\
\nfs^c_3(P^*, Q^\tr \leftand (R^* \leftand S^d)) &= \nfs^c_3(P^*, Q^\tr
  \leftand R^*) \leftand S^d
  \label{eq:nfsc18} \\
\nfs^c_3(P^*, Q^\tr \leftand (R^* \leftor S^c)) &= \nfs^c_1(P^*, Q^\tr)
  \leftand (R^* \leftor S^c).
  \label{eq:nfsc19} 
\end{align}

\begin{theorem}[Normal forms]
\label{thm:nfs}
For any $P \in \ST$, $\nfs(P)$ terminates, $\nfs(P) \in \SNF$ and 
\[\EqFSCL\vdash \nfs(P) = P.\]
\end{theorem}

In Appendix~\ref{app:nf} we first prove a number of lemmas showing that the
definitions $\nfs^n$ and $\nfs^c$ are correct and use those to prove the above
theorem. 
We have chosen to define normalisation by a function rather than by a 
rewriting system because this is more simple
and, if desirable, more appropriate for tool implementations.

\section{A completeness proof}
\label{sec:3} 
We analyse the $se$-images 
of $\ST$-terms and provide some results on uniqueness of such trees (Section~\ref{subsec:tree}). 
Then
we define an inverse function of $se$ (on the appropriate domain) with which we 
can complete the proof of the announced completeness theorem (Section~\ref{subsec:cpl}).

\subsection{Tree structure and decompositions}
\label{subsec:tree}
In Section~\ref{subsec:cpl} we will prove that on $\SNF$ we can invert the function $\SE$. 
To do this we need to
prove several structural properties of the trees in $se[\SNF]$,
the image of $\SE$. In the
definition of $\SE$ we can see how $\SE(P \leftand Q)$ is assembled from
$\SE(P)$ and $\SE(Q)$ and similarly for $\SE(P \leftor Q)$. To decompose 
trees in $se[\SNF]$ we will introduce some notation. 
The trees in the image of $\SE$ are all
finite binary trees over $A$ with leaves in $\{\tr, \fa\}$, i.e.,
$\SE[\ST] \subseteq \T$. We will now also consider the set $\Tone$ of binary
trees over $A$ with leaves in $\{\tr, \fa, \triangle\}$. 
The triangle will be used as a placeholder when composing or
decomposing trees. Replacement of the leaves of trees in $\Tone$ by trees in
$\T$ or $\Tone$ is defined analogous to replacement for trees in $\T$, adopting
the same notational conventions.
As a first example, we have by definition of $\SE$ that 
$\SE(P \leftand Q)$ can be
decomposed as
\begin{equation*}
\SE(P)\sub{\tr}{\triangle}\sub{\triangle}{\SE(Q)},
\end{equation*}
where $\SE(P)\sub{\tr}{\triangle} \in \Tone$ and $\SE(Q) \in \T$. We note that
this only works because the trees in the image of $\SE$, or in $\T$ in general,
do not contain any triangles. 
Of course, each tree $X\in\T$ has the \emph{trivial 
decomposition} that involves a replacement of the form $\sub{\triangle}Y$, namely 
\[\triangle\sub{\triangle}X.~\footnote{Also, for each $X\in\T$ it follows
 that $X=X\sub{\triangle}Y$ for any $Y\in\T$, but
 we do not consider $X\sub{\triangle}Y$ to be a `decomposition' of $X$ in
 this case.}
\]

\bigskip

We start with some simple properties
of the $\SE$-images of \tr-terms, \fa-terms, and
$*$-terms.

\newpage
\begin{lemma}[Leaf occurrences]
\label{lem:TF}~\textup{
\begin{enumerate}\setlength\itemsep{-1mm}
\item 
\emph{For any \tr-term $P$, $\SE(P)$ contains  $\tr$, but not $\fa$},
\item
\emph{For any \fa-term $P$, $\SE(P)$ contains  $\fa$, but not $\tr$},
\item
\emph{For any
$*$-term $P$, $\SE(P)$ contains both $\tr$ and $\fa$.}
\end{enumerate}}
\end{lemma}

\begin{proof}
By induction on the structure of $P$. A proof of the first two statements is trivial.
For the third statement, if $P$ is an $\ell$-term, we find that by definition
of the grammar of $P$ that
one branch from the root of $\SE(P)$ will only contain $\tr$ and not $\fa$,
and for the other branch this is the other way around. 

For the induction we have to consider both 
$\SE(P_1 \leftand P_2)$ and $\SE(P_1\leftor P_2)$.
Consider $\SE(P_1 \leftand P_2)$, which equals by definition $se(P_1)\sub{\tr}{se(P_2)}$.
By induction, both $se(P_1)$ and $se(P_2)$ contain both \tr\ and \fa,
so $\SE(P_1 \leftand P_2)$ contains both \tr\ and \fa.
The case $\SE(P_1\leftor P_2)$ can be dealt with in a similar way.
\end{proof}

Decompositions of the $se$-image
of $*$-terms turn out to be crucial in our approach. As an example, the $se$-image
of the $*$-term 
\[(P^\ell\leftor Q^\ell)\leftand R^\ell\quad\text{with}\quad 
P^\ell=((a\leftand\tr)\leftor\fa), ~Q^\ell=((b\leftand\tr)\leftor\fa), 
~R^\ell=((c\leftand\tr)\leftor\fa)\]
can be decomposed as $X_1\sub{\triangle}Y$ with $X_1\in\Tone$ as follows:
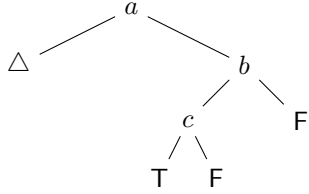
\[
\begin{tikzpicture}[%
level distance=7.5mm,
level 1/.style={sibling distance=30mm},
level 2/.style={sibling distance=15mm},
level 3/.style={sibling distance=7.5mm}
]
\node (a) {$a$}
  child {node (b1) {$\triangle$}
  }
  child {node (b2) {$b$}
    child {node (c3) {$c$}
      child {node (d5) {$\tr$}} 
      child {node (d6) {$\fa$}}
    }
    child {node (c4) {$\fa$}
    }
  };
\end{tikzpicture}
\]
and $Y=se(R^\ell)$, thus $Y=\tr\unlhd c\unrhd\fa$, and
two other decompositions are $X_2\sub{\triangle}Y=
X_3\sub{\triangle}Y$ with $X_2,X_3\in\Tone$ as follows:
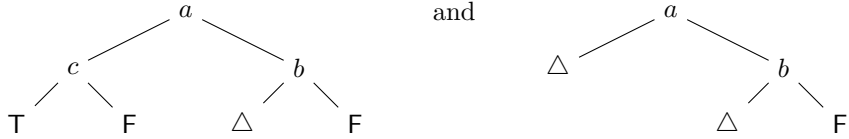
\[
\begin{tikzpicture}[%
level distance=7.5mm,
level 1/.style={sibling distance=30mm},
level 2/.style={sibling distance=15mm}
]
\node (a) {$a$}
  child {node (b1) {$c$}
    child {node (c1) {$\tr$}
    }
    child {node (c2) {$\fa$}
    }
  }
  child {node (b2) {$b$}
    child {node (c3) {$\triangle$}
    }
    child {node (c4) {$\fa$}
    }
  };
\end{tikzpicture}
\qquad\raisebox{16.6mm}{\text{ and }}\qquad
\begin{tikzpicture}[%
level distance=7.5mm,
level 1/.style={sibling distance=30mm},
level 2/.style={sibling distance=15mm}
]
\node (a) {$a$}
  child {node (b1) {$\triangle$}
  }
  child {node (b2) {$b$}
    child {node (c3) {$\triangle$}
    }
    child {node (c4) {$\fa$}
    }
  };
\end{tikzpicture}
\]
Observe that the first two decompositions have the property that $Y$ is a 
subtree of $X_1$ and $X_2$, respectively.
Furthermore, observe that $X_3=se(P^\ell\leftor Q^\ell)\sub{\tr}{\triangle}$, and 
hence that this decomposition agrees with the definition of the function $se$.
When we want to express that a certain decomposition $X\sub{\triangle}Y$ has the property
that $Y$ is not a subtree of $X$, we say that $X\sub{\triangle}Y$ 
is a \emph{strict decomposition}. Finally observe that each of these
decompositions satisfies the property that $X_i$ contains \tr\ or \fa, which
is a general property of decompositions of $*$-terms and a consequence of 
Lemma~\ref{lem:snondectf} (see below). The following lemma
provides the $\SE$-image of the rightmost $\ell$-term in a $*$-term as a witness.

\begin{lemma}[Witness decomposition]
\label{lem:sperttf}
For all $*$-terms $P$, $\SE(P)$ can be decomposed as $X\sub{\triangle}{Y}$ with $X
\in \Tone$ and $Y \in \T$ such that $X$ contains $\triangle$ and $Y = \SE(R)$ for
the rightmost $\ell$-term $R$ in $P$. 
Note that $X$ may be $\triangle$. 
\\[1mm]
We will refer to $Y$ as \textbf{the witness} for this lemma for $P$.
\end{lemma}

\begin{proof}
By induction on the number of $\ell$-terms in $P$. 
In the base case $P$ is an $\ell$-term and $\SE(P) =
\triangle\sub{\triangle}{\SE(P)}$ is the desired decomposition. 
For the induction we have to consider both $\SE(P \leftand
Q)$ and $\SE(P \leftor Q)$.

We start with $\SE(P \leftand Q)$ and let $X\sub{\triangle}{Y}$ be the
decomposition for $\SE(Q)$ which we have by induction hypothesis, so
$Y$ is the witness for this lemma for $Q$ and the $se$-image of its
rightmost $\ell$-term, say $R$. Since by
definition of $\SE$ on ${\leftand}$ we have
$\SE(P \leftand Q) = \SE(P)\sub{\tr}{\SE(Q)}$
we also have
\begin{equation*}
\SE(P \leftand Q) = \SE(P)\sub{\tr}{X\sub{\triangle}{Y}} =
\SE(P)\sub{\tr}{X}\sub{\triangle}{Y}.
\end{equation*}
The last equality is due to the fact that $\SE(P)$ does not contain any triangles.
This gives our desired decomposition: $\SE(P)\sub{\tr}{X}$ contains $\triangle$ because
$\SE(P)$ contains \tr\ (Lemma~\ref{lem:TF}) and $X$ contains $\triangle$,
and $Y$ is the $se$-image of the 
rightmost $\ell$-term $R$ in $P\leftand Q$.

The case for $\SE(P \leftor Q)$ is analogous.
\end{proof}

The following lemma illustrates another structural property of trees in the
image of $*$-terms under $\SE$, namely that each non-trivial decomposition
$X\sub{\triangle}{Y}$ of a $*$-term has the property that at least one of \tr\ and \fa\
occurs in $X$.

\begin{lemma}[Non-decomposition]
\label{lem:snondectf}
There is no $*$-term $P$ such that $\SE(P)$ can be decomposed as
$X\sub{\triangle}{Y}$ with $X \in \Tone$ and $Y \in \T$, where $X \neq \triangle$ and $X$
contains $\triangle$, but not $\tr$ or $\fa$.
\end{lemma}

\begin{proof}
By induction on the number of $\ell$-terms in $P$.
Let $P$ be a single $\ell$-term. When we analyse the grammar of $P$ we find that
one branch from the root of $\SE(P)$ only contains $\tr$ and not $\fa$,
and the other way around for the other branch. 
Hence if $\SE(P) = X\sub{\triangle}{Y}$ and $X$ does not contain \tr\ or \fa,
then $Y$ contains occurrences of both $\tr$ and $\fa$. Hence, $Y$ must contain the
root and $X = \triangle$.

For the induction we
assume that the lemma holds for all $*$-terms that contain fewer $\ell$-terms than $P
\leftand Q$ and $P \leftor Q$.
We start with the case for $\SE(P \leftand Q)$. Towards a contradiction, suppose 
that for some $*$-terms $P$ and $Q$,
\begin{equation}
\label{eq:sub}
\SE(P\leftand Q) = X\sub{\triangle}{Y}
\end{equation}
with $X \neq \triangle$ and $X$ not containing
any occurrences of $\tr$ or $\fa$. 
Let $Z$ be the witness of Lemma
\ref{lem:sperttf} for $P$
(so one branch of the root of $Z$ contains only \fa-leaves,
and the other only \tr-leaves). Observe that $\SE(P \leftand Q)$ has one or
more occurrences of the subtree
\[Z\sub{\tr}{\SE(Q)}.\]
The interest of this observation is that one branch of the root of this 
subtree contains only \fa, and the other branch contains both \tr\ and \fa\
(because $se(Q)$ does by Lemma~\ref{lem:TF}). It follows that all occurrences of 
$Z\sub{\tr}{\SE(Q)}$
in $se(P\leftand Q)$ are subtrees in $Y$ after being substituted in $X$:
\begin{itemize}\setlength\itemsep{-1mm}
\item
Because $X$ does not contain \tr\ and \fa, Lemma~\ref{lem:TF} and \eqref{eq:sub}
imply that $Y$ contains both \tr\ and \fa.
\item
Assume there is an occurrence of 
$Z\sub{\tr}{\SE(Q)}$ in $X\sub{\triangle}{Y}$ that has its root in $X$.
Hence the parts of the two branches from this root node that are in $X$
must have $\triangle$ as their leaves. 
For the branch that only has \fa-leaves this implies that 
$Y$ does not contain \tr, which is a contradiction. 
\end{itemize}
So, $Y$ contains at least one occurrence of $Z\sub{\tr}{\SE(Q)}$, hence
\begin{equation}
\label{JAB}
\text{$se(Q)$ is a \emph{proper} subtree of $Y$.}
\end{equation}
This implies that \emph{each} occurrence of $se(Q)$ in $se(P\leftand Q)$
is an occurrence in $Y$ (after being substituted): if this were
not the case, the root of $se(Q)$ occurs also in $X$
and the parts of the two branches from this node that are in $X$
must have $\triangle$ as their leaves, which implies that $Y$ after being substituted
in $X$ is a proper subtree of $se(Q)$. 
By~\eqref{JAB} this implies 
that $se(Q)$ is a proper subtree of itself, which is a contradiction.

Because each occurrence of $se(Q)$ in $se(P\leftand Q)= X\sub{\triangle}{Y}$
is an occurrence in $Y$ (after being substituted) and because 
$se(P\leftand Q)=se(P)\sub{\tr}{se(Q)}$, 
it follows that $\SE(P) = X\sub{\triangle}{V}$ where $V$ is obtained from $Y$ by
replacing all occurrences of the subtree ${\SE(Q)}$ by \tr. But this
violates the induction hypothesis. 
This concludes the induction step for the case of $se(P\leftand Q)$. 

A proof for the case $\SE(P \leftor Q)$ is symmetric.
\end{proof}

We now arrive at two crucial definitions concerning decompositions. When
considering $*$-terms, we already know that $\SE(P \leftand Q)$ can be
decomposed as
\begin{equation*}
\SE(P)\sub{\tr}{\triangle}\sub{\triangle}{\SE(Q)}.
\end{equation*}
Our goal now is to give a definition for a kind of decomposition so that this
is the only such decomposition for $\SE(P \leftand Q)$. We also ensure that
$\SE(P \leftor Q)$ does not have a decomposition of that kind, so that we can
distinguish $\SE(P \leftand Q)$ from $\SE(P \leftor Q)$. Similarly, we need
to define another kind of decomposition so that $\SE(P \leftor Q)$ can only be
decomposed as 
\begin{equation*}
\SE(P)\sub{\fa}{\triangle}\sub{\triangle}{\SE(Q)}
\end{equation*}
and that $\SE(P \leftand Q)$ does not have a decomposition of that kind.

\begin{definition}
\label{def:candidate}
The pair $(Y, Z) \in \Tone \times \T$ is a \textbf{candidate conjunction
decomposition (ccd)} of $X \in \T$, if
\begin{itemize}\setlength\itemsep{-1mm}
\item[$\bullet$] $X = Y\sub{\triangle}{Z}$,
\item[$\bullet$] $Y$ contains $\triangle$,
\item[$\bullet$] $Y$ contains $\fa$, but not $\tr$, and
\item[$\bullet$] $Z$ contains both $\tr$ and $\fa$.
\end{itemize}
Similarly, $(Y, Z)$ is a \textbf{candidate disjunction decomposition (cdd)} of
$X$, if
\begin{itemize}\setlength\itemsep{-1mm}
\item[$\bullet$] $X = Y\sub{\triangle}{Z}$,
\item[$\bullet$] $Y$ contains $\triangle$,
\item[$\bullet$] $Y$ contains $\tr$, but not $\fa$, and
\item[$\bullet$] $Z$ contains both $\tr$ and $\fa$.
\end{itemize}
\end{definition}

Observe that any ccd or cdd $(Y, Z)$ is \emph{strict} because $Z$
contains both
$\tr$ and \fa, and thus cannot be a subtree of $Y$.
A first, crucial property of ccd's and cdd's is the following connection
with $se$-images of $*$-terms.

\begin{lemma}
\label{lem:nocdd}
For any $*$-term $P \leftand Q$,
$\SE(P \leftand Q)$ has no cdd. 
Similarly, for any $*$-term $P \leftor Q$, 
$\SE(P \leftor Q)$ has no ccd.
\end{lemma}

\begin{proof}
We first treat the case for $P \leftand Q$,
so $P \in P^*$ and $Q \in P^d$.
Towards a contradiction, suppose that $(Y, Z)$ is a cdd of $\SE(P \leftand Q)$. 
Let $Z'$ be the witness of Lemma~\ref{lem:sperttf} for $P$.
Observe that $\SE(P \leftand Q)$ has one or
more occurrences of the subtree
\[Z'\sub{\tr}{\SE(Q)}.\]
It follows that all occurrences
of $Z'\sub{\tr}{se(Q)}$ in $se(P\leftand Q)$ are subtrees in $Z$ after being 
substituted in $Y$, which
can be argued in a similar way as in the proof of Lemma~\ref{lem:snondectf}:
\begin{itemize}\setlength\itemsep{-1mm}
\item Assume there is an occurrence of
$Z'\sub{\tr}{se(Q)}$ in $Y\sub{\triangle}{Z}$ that has its root in $Y$. Following
the branch from this node that only has \fa-leaves and that leads in $Y$
to one or more $\triangle$-leaves, this implies 
that $Z$ does not contain \tr, which is a contradiction by definition of a cdd.
\end{itemize}
So, $Z$ contains at least one occurrence of $Z'\sub{\tr}{se(Q)}$. This implies 
that \emph{each} occurrence of $se(Q)$ in $se(P\leftand Q)$ is an occurrence in 
$Z$ (after being substituted): if this were
not the case, the root of $se(Q)$ occurs in $Y$ and this
implies that $se(Q)$ is a proper subtree of itself, which is a contradiction.
By definition of $se$, all the occurrences of $\tr$ in $\SE(P \leftand Q)$
are in occurrences of the subtree $\SE(Q)$. 
Because $Y$ does not contain the root of an $se(Q)$-occurrence,
$Y$ does not contain any occurrences of
$\tr$, so $(Y,Z)$ is not a cdd of $\SE(P \leftand Q)$.  
A proof
for the case $\SE(P \leftor Q)$ is symmetric.
\end{proof}

However, the ccd and cdd are not necessarily the decompositions we are
looking for, because, for example, $\SE((P \leftand Q) \leftand R)$ has a ccd
\[(\SE(P)\sub{\tr}{\triangle}, \SE(Q \leftand R)),\] 
while the decomposition we
need to reconstruct the constituents of a $*$-term is
\[(\SE(P \leftand Q)\sub{\tr}{\triangle}, \SE(R)).\]
A more intricate example of a ccd $(Y, Z)$ 
that does not produce the constituents of a $*$-term
is this pair of trees $Y$ and $Z$: 
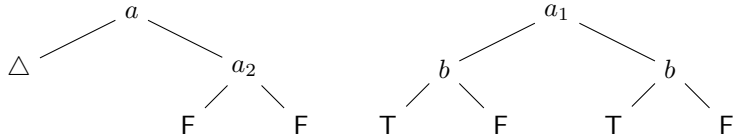
\[
\begin{tikzpicture}[%
level distance=7.5mm,
level 1/.style={sibling distance=30mm},
level 2/.style={sibling distance=15mm},
level 3/.style={sibling distance=7.5mm}
]
\node (a) {$a$}
  child {node (b1) {$\triangle$}
  }
  child {node (b2) {$a_2$}
    child {node (c3) {$\fa$}
    }
    child {node (c4) {$\fa$}
          }
  };
\end{tikzpicture}
\qquad
\begin{tikzpicture}[%
level distance=7.5mm,
level 1/.style={sibling distance=30mm},
level 2/.style={sibling distance=15mm},
level 3/.style={sibling distance=7.5mm}
]
\node (a) {$a_1$}
  child {node (b1) {$b$}
    child {node (c1) {$\tr$}
    }
    child {node (c2) {$\fa$}
    }
  }
  child {node (b2) {$b$}
    child {node (c3) {$\tr$}
    }
    child {node (c4) {$\fa$}
          }
  };
\end{tikzpicture}
\]
It is clear that $(Y,Z)$ is a ccd of $se(P^\ell\leftand Q^\ell)$ with $P^\ell$ 
and $Q^\ell$ these $\ell$-terms:
\[
P^\ell=(a\leftand ((a_1\leftand\tr)\leftor\tr))\leftor((a_2\leftor\fa)\leftand\fa),
\quad
Q^\ell=(b\leftand \tr)\leftor\fa.
\]
Therefore we refine Definition~\ref{def:candidate} to obtain the decompositions we seek.

\begin{definition}
The pair $(Y, Z) \in \Tone \times \T$ is a \textbf{conjunction decomposition
(cd)} of $X \in \T$, if it is a ccd of $X$ and there is no other ccd $(Y', Z')$
of $X$ where the depth of $Z'$ is smaller than that of $Z$. 

Similarly, the pair $(Y, Z) \in \Tone \times \T$
is a \textbf{disjunction decomposition (dd)} of $X$, if it is a cdd of $X$ and
there is no other cdd $(Y', Z')$ of $X$ where the depth of $Z'$ is smaller than
that of $Z$.
\end{definition}

\begin{theorem}
\label{thm:scddd}
For any $*$-term $P \leftand Q$, i.e., with $P \in P^*$ and $Q \in P^d$,
$\SE(P \leftand Q)$ has the unique cd
\begin{equation*}
(\SE(P)\sub{\tr}{\triangle}, \SE(Q))
\end{equation*}
and no dd. For any $*$-term $P \leftor Q$, i.e., with $P \in P^*$ and $Q \in
P^c$, $\SE(P \leftor Q)$ has no cd and its unique dd is
\begin{equation*}
(\SE(P)\sub{\fa}{\triangle}, \SE(Q)).
\end{equation*}
\end{theorem}

\begin{proof}
By simultaneous induction on the number of $\ell$-terms
in $P \leftand Q$ and $P \leftor Q$. 

In the basis we have to consider, for $\ell$-terms
$P^\ell$ and $Q^\ell$, the terms $P^\ell \leftand Q^\ell$ and $P^\ell \leftor 
Q^\ell$. By symmetry, it is sufficient to consider the first case. 
By definition of a ccd and Lemma~\ref{lem:TF},
$(\SE(P^\ell)\sub{\tr}{\triangle}, \SE(Q^\ell))$ is a ccd of
$\SE(P^\ell \leftand Q^\ell)$. Furthermore observe that the smallest subtree
in $se(P^\ell\leftand Q^\ell)$ that contains both \tr\ and \fa\ is $se(Q^\ell)$. 
Therefore $(\SE(P^\ell)\sub{\tr}{\triangle}, \SE(Q^\ell))$ is the \emph{unique} cd of
$\SE(P^\ell\leftand Q^\ell)$.
Now for the dd. It suffices to show that there is no cdd of 
$\SE(P^\ell \leftand Q^\ell)$ and this follows from Lemma~\ref{lem:nocdd}. 

For the induction we assume that the
theorem holds for all $*$-terms with fewer $\ell$-terms than $P \leftand Q$ and
$P \leftor Q$. 
We will first treat the case for $P \leftand Q$ and show that
$(\SE(P)\sub{\tr}{\triangle}, \SE(Q))$ is the unique cd of $\SE(P \leftand Q)$. 
In this case, observe that for any other ccd $(Y, Z)$ either $Z$ 
is a proper subtree of
$\SE(Q)$, or vice versa: if this were not the case, then there are occurrences of
$Z$ and $se(Q)$ in $Y\sub{\triangle}Z=se(P\leftand Q)$ 
that are disjoint and at least one of the following cases applies:
\begin{itemize}\setlength\itemsep{-1mm}
\item $Y$ contains an occurrence of $se(Q)$, and hence of \tr, which is a 
contradiction.
\item $se(P)\sub{\tr}{\triangle}$ contains an occurrence of $Z$, and hence of \tr,
which is a contradiction.
\end{itemize}
Hence, by definition of a cd
it suffices to show that there is no ccd $(Y, Z)$ where $Z$ is
a proper subtree of $\SE(Q)$. 
Towards a contradiction, suppose that such a ccd $(Y,
Z)$ does exist.
By definition of $*$-terms $Q$ is either an $\ell$-term or a disjunction. 
\begin{itemize}\setlength\itemsep{-1mm}
\item
If
$Q$ is an $\ell$-term and $Z$ a proper subtree of $\SE(Q)$, then $Z$
does not contain both $\tr$ and $\fa$ because one branch from the root of 
$\SE(Q)$ will only contain $\tr$ and not $\fa$,
and the other branch vice versa. 
Therefore $(\SE(P)\sub{\tr}{\triangle}, \SE(Q))$ is the \emph{unique} cd of
$\SE(P \leftand Q)$.
\item
If $Q$ is a disjunction and $Z$ a proper subtree of $\SE(Q)$, then we can
decompose $\SE(Q)$ as $\SE(Q) = U\sub{\triangle}{Z}$ for some $U \in \Tone$ that
contains but is not equal to $\triangle$ and such that $U\sub{\triangle}{Z}$ is strict, 
i.e., $Z$ is not a subtree of $U$. By Lemma~\ref{lem:snondectf} this implies
that $U$ contains either $\tr$ or $\fa$. 
\begin{itemize}\setlength\itemsep{-1mm}
\item
If $U$ contains $\tr$, then so
does $Y$, because 
$Y = \SE(P)\sub{\tr}{U}$, which is the case
because 
\begin{align*}
Y\sub{\triangle}Z&=se(P\leftand Q)\\
&=se(P)\sub{\tr}{U\sub{\triangle}Z}\\
&=se(P)\sub{\tr}{U}\sub{\triangle}Z,
\end{align*}
and the only way in which $Y\ne\SE(P)\sub{\tr}{U}$ is possible is that 
$U$ contains an occurrence of $Z$, which is excluded because
$U\sub{\triangle}{Z}$ is strict. Because $Y$ contains an occurrence of \tr,
$(Y, Z)$ is not a ccd of $\SE(P \leftand Q)$.
\item
If $U$ only contains
$\fa$ then $(U, Z)$ is a ccd of $\SE(Q)$ which violates the induction
hypothesis. 
\end{itemize}
Therefore $(\SE(P)\sub{\tr}{\triangle}, \SE(Q))$ is the \emph{unique}
cd of $\SE(P \leftand Q)$.
\end{itemize}
Now for the dd. By Lemma~\ref{lem:nocdd} there is no cdd of $\SE(P \leftand
Q)$, so there is neither a dd of $\SE(P \leftand Q)$. 
A proof
for the case $\SE(P \leftor Q)$ is symmetric.
\end{proof}

At this point we have the tools necessary to invert $\SE$ on $*$-terms, at
least down to the level of $\ell$-terms. We can easily detect if a tree in the
image of $\SE$ is in the image of $P^\ell$, because all leaves to the left of
the root are one truth value, while all the leaves to the right are the other.
To invert $\SE$ on $\tr$-$*$-terms we still need to be able to reconstruct
$\SE(P^\tr)$ and $\SE(Q^*)$ from $\SE(P^\tr \leftand Q^*)$. To this end we
define a $\tr$-$*$-decomposition, and as with cd's and dd's we first define a
candidate $\tr$-$*$-decomposition.

\begin{definition}
The pair $(Y, Z) \in \Tone \times \T$ is a \textbf{candidate
$\tr$-$*$-decomposition (ctsd)} of $X \in \T$, if 
\begin{itemize}\setlength\itemsep{-1mm}
\item[$\bullet$]
$X = Y\sub{\triangle}{Z}$, 
\item[$\bullet$] 
$Y$
does not contain $\tr$ or $\fa$,
\item[$\bullet$] $Z$ contains both \tr\ and \fa,
\end{itemize}
and there is no decomposition $(U, V) \in
\Tone \times \T$ of $Z$ such that
\begin{itemize}\setlength\itemsep{-1mm}
\item[$\bullet$] $Z = U\sub{\triangle}{V}$,
\item[$\bullet$] $U$ contains $\triangle$,
\item[$\bullet$] $U \neq \triangle$, and
\item[$\bullet$] $U$ contains neither $\tr$ nor $\fa$.
\end{itemize}
\end{definition}

However, this is not necessarily the decomposition we seek in this case.
Consider for example the $\tr$-term $P^\tr$ with the following semantics:
\begin{center}
\begin{tikzpicture}[%
level distance=7.5mm,
level 1/.style={sibling distance=30mm},
level 2/.style={sibling distance=15mm},
level 3/.style={sibling distance=7.5mm}
]
\node (a) {$a$}
  child {node (b1) {$b$}
    child {node (c1) {$c$}
      child {node (d1) {$\tr$}} 
      child {node (d2) {$\tr$}}
    }
    child {node (c2) {$d$}
      child {node (d3) {$\tr$}} 
      child {node (d4) {$\tr$}}
    }
  }
  child {node (b2) {$b$}
    child {node (c3) {$c$}
      child {node (d5) {$\tr$}} 
      child {node (d6) {$\tr$}}
    }
    child {node (c4) {$d$}
      child {node (d7) {$\tr$}} 
      child {node (d8) {$\tr$}}
    }
  };
\end{tikzpicture}
\end{center}
and observe that
$\SE(P^\tr \leftand Q^*)$ has a ctsd
\begin{equation*}
(\triangle \tlef a \trig \triangle, (\SE(Q^*) \tlef c \trig \SE(Q^*)) \tlef b \trig
(\SE(Q^*) \tlef d \trig \SE(Q^*))).
\end{equation*}
But the decomposition we seek is $(\SE(P^\tr)\sub{\tr}{\triangle}, \SE(Q^*))$.
Hence we will refine the above definition to aid in the theorem below.

\begin{definition}
The pair $(Y, Z) \in \Tone \times \T$ is a \textbf{$\tr$-$*$-decomposition
(tsd)} of $X \in \T$, if it is a ctsd of $X$ and there is no other ctsd $(Y',
Z')$ of $X$ where the depth of $Z'$ is smaller than that of $Z$.
\end{definition}

\begin{theorem}
\label{thm:stsd}
For any $\tr$-term $P$ and $*$-term $Q$ the unique tsd of $\SE(P \leftand
Q)$ is
\begin{equation*}
(\SE(P)\sub{\tr}{\triangle}, \SE(Q)).
\end{equation*}
\end{theorem}
\begin{proof}
First observe that $(\SE(P)\sub{\tr}{\triangle}, \SE(Q))$ is a ctsd because by
definition $\SE(P)\sub{\tr}{\SE(Q)} = \SE(P
\leftand Q)$ and $\SE(Q)$ is non-decomposable by Lemma~\ref{lem:snondectf}.

Towards a contradiction, suppose there exists a ctsd $(Y, Z)$ such that the depth of
$Z$ is smaller than that of $\SE(Q)$. Now either $Z$ is a proper subtree of
$\SE(Q)$, or vice versa, for otherwise there would be occurrences of
$Z$ and $se(Q)$ in $Y\sub{\triangle}Z=\SE(P)\sub{\tr}{\SE(Q)}$ 
that are disjoint and at least one of the following cases applies:
\begin{itemize}\setlength\itemsep{-1mm}
\item $Y$ contains an occurrence of $se(Q)$, and hence of \tr\ and \fa, which is a 
contradiction.
\item $se(P)\sub{\tr}{\triangle}$ contains an occurrence of $Z$, and hence of \tr\ and \fa,
which is a contradiction.
\end{itemize}
By definition of a tsd if suffices to only consider the case that 
$Z$ is a proper subtree of $se(Q)$. If this is the case, then
$\SE(Q) = U\sub{\triangle}{Z}$ for some $U \in \Tone$ that is not equal
to $\triangle$ and does not contain $\tr$ or $\fa$ (because then $Y$ would
too). But this violates Lemma~\ref{lem:snondectf}, which states that no such
decomposition exists. 
Hence, $(\SE(P)\sub{\tr}{\triangle}, \SE(Q))$ is the \emph{unique} tsd of 
$\SE(P \leftand Q)$.
\end{proof}

\subsection{Defining an inverse and proving completeness}
\label{subsec:cpl}
The two decomposition theorems from the previous section
enable us to prove the intermediate result that we used in our completeness proof
for $\FSCL$. We define three auxiliary functions to aid in our definition of the
inverse of $\SE$ on $\SNF$. Let 
\[\cd : \T \to \Tone \times \T\]
be the function
that returns the conjunction decomposition of its argument, $\dd$ of the same
type its disjunction decomposition and $\tsd$, also of the same type, its
$\tr$-$*$-decomposition. Naturally, these functions are undefined when their
argument does not have a decomposition of the specified type. Each of these
functions returns a pair and we will use $\cd_1$ ($\dd_1$, $\tsd_1$) to denote
the first element of this pair and $\cd_2$ ($\dd_2$, $\tsd_2$) to denote the
second element.

We define $\invs: \T \to \ST$ using the functions $\invs^\tr: \T \to \ST$ for
inverting trees in the image of $\tr$-terms and $\invs^\fa$, $\invs^\ell$
and $\invs^*$ of the same type for inverting trees in the image of
$\fa$-terms, $\ell$-terms and $*$-terms, respectively. These functions are
defined as follows.
\begin{align}
\label{eq:g1}
\invs^\tr(X) &=
  \begin{cases}
    \tr
      &\text{if $X = \tr$,} \\
    (a \leftand \invs^\tr(Y)) \leftor \invs^\tr(Z) 
      &\text{if $X = Y \tlef a \trig Z$.}
  \end{cases} \\ \displaybreak[0]
\label{eq:g2}
\invs^\fa(X) &=
  \begin{cases}
    \fa
      &\text{if $X = \fa$,} \\
    (a \leftor \invs^\fa(Z)) \leftand \invs^\fa(Y)
      &\text{if $X = Y \tlef a \trig Z$.}
  \end{cases} \\ \displaybreak[0]
\label{eq:g3}
\invs^\ell(X) &=
  \begin{cases}
    (a \leftand \invs^\tr(Y)) \leftor \invs^\fa(Z) 
      &\text{if $X = Y \tlef a \trig Z$ and $Y$ only has $\tr$-leaves,} \\
    (\neg a \leftand \invs^\tr(Z)) \leftor \invs^\fa(Y)
      &\text{if $X = Y \tlef a \trig Z$ and $Z$ only has $\tr$-leaves.}
  \end{cases} \\ \displaybreak[0]
\label{eq:g4}
\invs^*(X) &=
  \begin{cases}
    \invs^*(\cd_1(X)\sub{\triangle}{\tr}) \leftand \invs^*(\cd_2(X))
      &\text{if $X$ has a cd,} \\
    \invs^*(\dd_1(X)\sub{\triangle}{\fa}) \leftor \invs^*(\dd_2(X))
      &\text{if $X$ has a dd,} \\
    \invs^\ell(X)
      &\text{otherwise.}
  \end{cases} \\ \displaybreak[0]
\label{eq:g5}
\invs(X) &=
  \begin{cases}
    \invs^\tr(X)
      &\text{if $X$ has only $\tr$-leaves,} \\
    \invs^\fa(X)
      &\text{if $X$ has only $\fa$-leaves,} \\
    \invs^\tr(\tsd_1(X)\sub{\triangle}{\tr}) \leftand \invs^*(\tsd_2(X))
      &\text{otherwise.}
  \end{cases}
\end{align}

We use the symbol $\equiv$ to denote `syntactic equivalence' and we have
the following result on our normal forms.
\begin{theorem}
\label{thm:sclinv}
For all $P \in \SNF$, $\invs(\SE(P)) \equiv P$.
\end{theorem}

\begin{proof}
We first prove that for all $\tr$-terms $P$, $\invs^\tr(\SE(P)) \equiv P$,
by induction on $P$. In the base case $P \equiv \tr$ and we have
by~\eqref{eq:g1} that
$\invs^\tr(\SE(P)) \equiv \invs^\tr(\tr) \equiv \tr \equiv P$. For the
inductive case we have $P \equiv (a \leftand Q^\tr) \leftor R^\tr$ and
\begin{align*}
\invs^\tr(\SE(P)) &\equiv \invs^\tr(\SE(Q^\tr) \tlef a \trig
  \SE(R^\tr))
&&\text{by definition of $\SE$} \\
&\equiv (a \leftand \invs^\tr(\SE(Q^\tr))) \leftor
  \invs^\tr(\SE(R^\tr))
&&\text{by \eqref{eq:g1}} \\
&\equiv (a \leftand Q^\tr) \leftor R^\tr
&&\text{by induction hypothesis} \\
&\equiv P.
\end{align*}
In a similar way it follows by~\eqref{eq:g2} that for all $\fa$-terms $P$, 
$\invs^\fa(\SE(P)) \equiv P$.

Next we check that for all $\ell$-terms $P$, $\invs^\ell(\SE(P)) \equiv P$.  We
observe that either $P \equiv (a \leftand Q^\tr) \leftor R^\fa$ or $P
\equiv (\neg a \leftand Q^\tr) \leftor R^\fa$. In the first case we have
\begin{align*}
\invs^\ell(\SE(P)) &\equiv \invs^\ell(\SE(Q^\tr) \tlef a \trig \SE(R^\fa))
&&\text{by definition of $\SE$} \\
&\equiv (a \leftand \invs^\tr(\SE(Q^\tr))) \leftor
  \invs^\fa(\SE(R^\fa))
&&\text{by~\eqref{eq:g3}, first case} \\
&\equiv (a \leftand Q^\tr) \leftor R^\fa
&&\text{as shown above} \\
&\equiv P.
\end{align*}
The second case follows in a similar way.

We now prove that for all $*$-terms $P$, $\invs^*(\SE(P)) \equiv P$, by
induction on the number of $\ell$-terms in $P$. In the base case we are
dealing with $\ell$-terms. Because an $\ell$-term has neither a cd nor a dd we
have $\invs^*(\SE(P)) \equiv \invs^\ell(\SE(P)) \equiv P$, 
where the first identity is by~\eqref{eq:g4} and the second identity was shown above. 
For the induction we have either $P \equiv Q \leftand R$ or $P \equiv Q \leftor R$.
In the first case note that by Theorem~\ref{thm:scddd}, $\SE(P)$ has a unique
cd and no dd.  So we have 
\begin{align*}
\invs^*(\SE(P)) &\equiv \invs^*(\cd_1(\SE(P))\sub{\triangle}{\tr}) \leftand
  \invs_*(\cd_2(\SE(P)))
&&\text{by~\eqref{eq:g4}} \\
&\equiv \invs^*(\SE(Q)) \leftand \invs^*(\SE(R))
&&\text{by Theorem~\ref{thm:scddd}} \\
&\equiv Q \leftand R
&&\text{by induction hypothesis} \\
&\equiv P.
\end{align*}
In the second case, again by Theorem~\ref{thm:scddd}, $P$ has a unique dd and 
no cd. So we have that
\begin{align*}
\invs^*(\SE(P)) &\equiv \invs^*(\dd_1(\SE(P))\sub{\triangle}{\fa}) \leftor
  \invs_*(\dd_2(\SE(P)))
&&\text{by~\eqref{eq:g4}} \\
&\equiv \invs^*(\SE(Q)) \leftor \invs^*(\SE(R))
&&\text{by Theorem~\ref{thm:scddd}} \\
&\equiv Q \leftor R
&&\text{by induction hypothesis} \\
&\equiv P.
\end{align*}

Finally, we prove the theorem's statement by making a case distinction on the
grammatical category of $P$. If $P$ is a $\tr$-term, then $\SE(P)$ has only
$\tr$-leaves and hence $\invs(\SE(P)) \equiv \invs^\tr(\SE(P)) \equiv P$,
where the first identity is by definition~\eqref{eq:g5}
of $\invs$ and the second identity was shown
above. If $P$ is an $\fa$-term, then $\SE(P)$ has only $\fa$-leaves and
hence $\invs(\SE(P)) \equiv \invs^\fa(\SE(P)) \equiv P$, where the first
identity is by definition~\eqref{eq:g5} of $\invs$ and the second one
was shown above. 
If $P$ is a $\tr$-$*$-term, then it has both $\tr$ and $\fa$-leaves and hence,
letting $P \equiv Q \leftand R$,
\begin{align*}
\invs(\SE(P)) &\equiv \invs^\tr(\tsd_1(\SE(P))\sub{\triangle}{\tr}) \leftand
  \invs^*(\tsd_2(\SE(P)))
&&\text{by \eqref{eq:g5}} \\
&\equiv \invs^\tr(\SE(Q)) \leftand \invs^*(\SE(R))
&&\text{by Theorem~\ref{thm:stsd}} \\
&\equiv Q \leftand R
&&\text{as shown above} \\
&\equiv P,
\end{align*}
which completes the proof.
\end{proof}

\begin{theorem}[Completeness of \SCLe\ for closed terms]
\label{thm:sclcpl}
For all $P,Q\in\SP$,
\[\SCLe\vdash  P=Q \quad\iff\quad
\Mse\models P=Q.\]
\end{theorem}

\begin{proof}
$\Rightarrow$ follows from Theorem~\ref{thm:sclsnd}.
For $\Leftarrow$, assume $\Mse\models P=Q$, thus $\SE(P) = \SE(Q)$. 
By Theorem~\ref{thm:nfs}, $P$ is derivably equal to an $\SNF$-term
$P'$, i.e., $\EqFSCL \vdash P = P'$, and  $Q$ is derivably equal to an
$\SNF$-term $Q'$, i.e., $\EqFSCL \vdash Q = Q'$. 
By Theorem~\ref{thm:sclsnd}, $\SE(P) = \SE(P')$ and $\SE(Q) = \SE(Q')$, 
so $\invs(se(P'))\equiv\invs(se(Q'))$. By 
Theorem~\ref{thm:sclinv} it follows that 
$P'\equiv Q'$ and hence $\EqFSCL \vdash P' = Q'$, and thus $\EqFSCL \vdash P = Q$.
\end{proof}

\section{Short-circuit logic, evaluation strategies, and side effects}
\label{sec:4}
We consider Hoare's \emph{conditional}, a ternary connective that can be used for defining 
the connectives of \SigSCL\ (Section~\ref{subsec:4.1}).
Then we recall the definition of free short-circuit logic (\FSCL) that was published earlier 
and give a partial solution of an open question (Section~\ref{subsec:4.2}).
Next we consider some other evaluation strategies leading to short-circuit logics that identify
more sequential propositional statements than \FSCL\ does (Section~\ref{subsec:4.3}).
Finally, we briefly discuss \emph{side effects} (Section~\ref{subsec:effects}).

\subsection{Hoare's conditional connective}
\label{subsec:4.1}
In 1985, Hoare introduced the \emph{conditional} (\cite{Hoa85}), a
ternary connective with notation
\[x\lef y \rig z.\]
A more common expression for the conditional $x\lef y\rig z$
is
\[
\texttt{if }y\texttt{ then }x\texttt{ else }z,
\]
which emphasises that $y$ is evaluated \emph{first}, and depending
on the outcome of this partial evaluation, either $x$ or $z$ is evaluated,
which determines the evaluation result. So, the evaluation strategy
prescribed by this form of if-then-else is a prime example of a sequential
evaluation strategy.
In order to reason algebraically with conditional expressions, 
Hoare's more `operator like' notation $x\lef y\rig z$ seems indispensable.
In~\cite{Hoa85} an equational axiomatisation of 
propositional logic is provided that only uses the conditional. Furthermore it is 
described how the binary connectives and negation are expressed in his set-up
(however, the sequential nature of the conditional's
evaluation is not discussed in this paper).
This axiomatisation over the signature $\SigCP=\{\_\lef\_\rig\_\,,\tr,\fa,a\mid a\in A\}$
consists of eleven axioms and includes the four
axioms in Table~\ref{CP}, and some more axioms, for example 
\[x\lef y\rig(z\lef y \rig w)=x\lef y\rig w.\]
The four axioms in Table~\ref{CP}, named CP (for Conditional Propositions),
establish a complete axiomatisation of \emph{free valuation congruence}
defined in~\cite{BP10}.

\begin{table}
{ \small
\centering
\rule{1\textwidth}{.4pt}
\begin{align*}
\label{cp1}
\tag{CP1} x \lef \tr \rig y &= x\\
\label{cp2}\tag{CP2}
x \lef \fa \rig y &= y\\
\label{cp3}\tag{CP3}
\tr \lef x \rig \fa  &= x\\
\label{cp4}\tag{CP4}
\qquad
    x \lef (y \lef z \rig u)\rig v &= 
	(x \lef y \rig v) \lef z \rig (x \lef u \rig v)
\end{align*}
\hrule
}
\caption{The set \CP\ of axioms for proposition algebra}
\label{CP}
\end{table}

By extending the definition of the function $se$ 
(Definition~\ref{def:se}) to closed terms over $\SigCP$ with
\[se(P\lef Q\rig R)=se(Q)[\tr\mapsto se(P),\fa\mapsto se(R)],\]
we can now characterise the completeness of \CP\ (mentioned above) by
\begin{equation}
\label{**}
\text{For all closed terms $P,Q$ over \SigCP, $\CP\vdash P=Q\iff se(P)=se(Q)$.}
\end{equation}
A simple and concise proof of this fact is recorded in~\cite[Thm.2.11]{BP15} 
and repeated in Appendix~\ref{app:cp}.

With the conditional connective and the constants \tr\ and \fa, 
the sequential connectives prescribing short-circuit evaluation are definable:
\begin{align}
\label{defneg}
\neg x &=\fa\lef x\rig \tr,\\
\label{defand}
x\leftand y &=y\lef x\rig\fa,\\
\label{defor}
x\leftor y &=\tr\lef x\rig y.
\end{align}
Observe that these equations satisfy the extended definition of 
the function $se$, that is,
$se(\neg P)=se(\fa\lef P\rig\tr)$, $se(P\leftand Q)=se(Q\lef P\rig\fa)$, and $se(P\leftor Q)=se(\tr\lef P\rig Q)$.

Thus, the axioms in Table~\ref{CP} combined with equations~\eqref{defneg}-\eqref{defor}, say 
$\CPandneg$, axiomatise
equality of evaluation trees for closed terms over the enriched signature
$\SigCP\cup\SigSCL$.

\subsection{The definition of free short-circuit logic}
\label{subsec:4.2}
In~\cite{BP12a,BPS13} a set-up is provided for defining short-circuit 
logics in a generic way with help of the
conditional by restricting the enriched language of $\CPandneg$ to the signature $\SigSCL$. 
The conditional connective is declared as a \emph{hidden operator}.

Intuitively, a
short-circuit logic is a logic that implies all
consequences of some \CP-axiomatisation that can be expressed in the signature
$\{\tr,\neg,\leftand\}$. 
The definition
below uses the export-operator ${\export}$ of \emph{Module
algebra}~\cite{BHK90} to express this in a concise way:
in module algebra, $S\export X$  is the operation that 
exports the signature $S$ from module $X$ while declaring 
other signature elements hidden. 

\begin{definition}
\label{def:SCL}
A \textbf{short-circuit logic}
is a logic that implies the consequences
of the module expression
\begin{align*}
\SCL=\{\tr,\neg,\leftand\}\export(&\CP\cup\{\eqref{defneg},\eqref{defand}\}).
\end{align*}
\end{definition}

As a first example, $~\SCL\vdash \neg\neg x=x$~ 
can be formally proved as follows:
\begin{align}
\nonumber
\neg \neg x&=\fa\lef(\fa\lef x\rig\tr)\rig\tr
&&\text{by \eqref{defneg}}\\
\nonumber
&=(\fa\lef\fa\rig\tr)\lef x\rig(\fa\lef\tr\rig\tr)
&&\text{by~\eqref{cp4}}\\
\nonumber
&=\tr\lef x\rig\fa
&&\text{by~\eqref{cp2}, \eqref{cp1}}\\
\label{eq:dns}
&=x.
&&\text{by~\eqref{cp3}}
\end{align}

The constant \fa\
and the connective ${\leftor}$
are not in the exported signature of \SCL,
but can be easily added to \SCL\ by their defining axioms,
just as was done in \SCLe:
\[\SCL\vdash \neg\tr=\fa\text{ because }
\CP\cup\{\eqref{defneg}\}\vdash\neg\tr=\fa\lef \tr\rig\tr
=\fa,\]
and $\SCL\vdash x\leftor y=\neg(\neg x\leftand
\neg y)$ because in $\CP\cup\{\eqref{defneg},\eqref{defand}\}$ one can derive
\begin{align*}
\neg(\neg x\leftand
\neg y)&=\fa\lef(\neg y\lef[\fa\lef x\rig\tr]\rig\fa)\rig\tr
&&\text{by definition}\\
&=\fa\lef(\fa\lef x\rig\neg y)\rig\tr
&&\text{by~\eqref{cp4}, \eqref{cp2}, \eqref{cp1}}\\
&=(\fa\lef\fa\rig\tr)\lef x\rig(\fa\lef\neg y\rig\tr)
&&\text{by~\eqref{cp4}}\\
&=\tr\lef x\rig y.
&&\text{by~\eqref{cp2}, \eqref{defneg}, \eqref{eq:dns}}
\end{align*}

In~\cite{BP12a,BPS13}, \emph{free} short-circuit logic is defined
as the least identifying short-circuit logic:

\begin{definition}
\label{def:FSCL}
\textbf{Free short-circuit logic $(\FSCL)$}
is the short-circuit logic that implies no other 
consequences than those of the module expression \SCL.
\end{definition}

An open problem posed in~\cite{BP12a} is to prove that
for all terms $s$ and $t$ over $\SigSCL$, 
$\SCLe\vdash s=t\iff \FreeSCL\vdash s=t$.
The following result solves this problem for closed terms only. 

\begin{theorem} 
\label{thm:openq}
For all $P,Q\in\SP$, $\SCLe\vdash P=Q \iff \FSCL\vdash P=Q$.
\end{theorem}
\begin{proof}
($\Rightarrow$) All axioms of $\SCLe$ are derivable in \FSCL. As an example,
we derive~\eqref{SCL9}:
\begin{align*}
(x\leftand\fa)\leftor y
&=\tr\lef (\fa\lef x\rig\fa)\rig y
&&\text{by definition}\\
&=(\tr\lef \fa\rig y)\lef x\rig(\tr\lef \fa\rig y)
&&\text{by~\eqref{cp4}}\\
&=(y\lef\tr\rig\fa)\lef x\rig(y\lef\tr\rig\fa)
&&\text{by~\eqref{cp1}, \eqref{cp2}}\\
&=y\lef (\tr\lef x\rig\tr)\rig \fa
&&\text{by~\eqref{cp4}}\\
&=(x\leftor \tr)\leftand y.
&&\text{by definition}
\end{align*}
($\Leftarrow$) Assume $\FSCL\vdash P=Q$. By the extended definition of $se$
we find with equivalence~\eqref{**} that
$se(P)=se(Q)$. By Theorem~\ref{thm:sclcpl}, $\SCLe\vdash P=Q$.
\end{proof}
\subsection{More short-circuit logics and evaluation strategies}
\label{subsec:4.3}
Following Definition~\ref{def:FSCL}, variants of \FSCL\ that identify more sequential propositions
can be easily defined.
As an example, adding in \SCL's definition to $\CP$ the two equation schemes
\begin{equation}
\label{rpaxioms}
(x\lef a\rig y) \lef a\rig z=(x\lef a\rig x)\lef a\rig z
\quad\text{and}\quad
x\lef a\rig (y\lef a\rig z)=x\lef a\rig (z\lef a\rig z)
\end{equation}
where $a$ ranges over $A$ defines \emph{repetition-proof} SCL (\RPSCL), in which
subsequent atomic evaluations of $a$ yield the same 
atomic evaluation results.
For example,
\begin{equation}
\label{rp:ex}
\RPSCL\vdash a\leftand (a\leftor x)=a\leftand a.
\end{equation}
For \RPSCL\ there exist natural examples (in Section~\ref{subsec:effects}, we sketch one briefly). 
Evaluation trees for \RPSCL\ can be defined by a transformation of $se$-trees
according to the axiom schemes~\eqref{rpaxioms}, see~\cite{BP15}.

For another example, adding in \SCL's definition to $\CP$ the two axioms
\begin{equation}
\label{stataxioms}
x\lef y\rig (z\lef u\rig(v\lef y\rig w))=x\lef y\rig (z\lef u\rig w)
\quad\text{and}\quad
\fa\lef x\rig \fa=\fa
\end{equation}
defines \emph{static} SCL (see~\cite{BPS13}),
which is a sequential form of propositional logic. 
Note that the first axiom in~\eqref{stataxioms}
and those of \CP\ (in Table~\ref{CP})
imply the axioms schemes~\eqref{rpaxioms} (set $u=\fa$ and $y=a$, $y=\fa\lef a\rig\tr$, respectively).

\bigskip

Another sequential evaluation strategy is so-called
\emph{full sequential evaluation}, which evaluates \emph{all} atoms in a compound
statement from left to right. We use the notations $x\fulland y$ and $x\fullor y$
for the connectives that prescribe full sequential evaluation. 
The setting with only full sequential connectives (thus, without short-circuit connectives)
can be called `free full sequential logic', and an axiomatisation is provided in~\cite{Daan}.
This axiomatisation also comprises axioms~\eqref{SCL1} and~\eqref{SCL3}, and a typical axiom is
\begin{equation}
\label{full}
x\fulland\fa = \fa\fulland x.
\end{equation} 
With the tool \emph{Prover9}~\cite{BirdBrain} it follows that~\eqref{SCL1} is
derivable, and with the tool \emph{Mace4}~\cite{BirdBrain}
it follows that the remaining axioms in~\cite{Daan} are independent (even if $|A|=1$).
Furthermore, both \eqref{SCL1} and~\eqref{SCL3} become derivable if the axiom $x\fulland\tr=x$ is 
replaced by $x\fullor\fa=x$, and the remaining axioms are again independent (even if $|A|=1$).

As is also noted in~\cite{Daan}, the `full sequential connectives' can be defined in terms of 
$\leftand$ and $\leftor$, and the constants \tr\ and \fa: 
\begin{align*}
x\fulland y &= (x\leftor (y\leftand\fa))\leftand y
\quad\text{and}\quad
x\fullor y = (x\leftand (y\leftor\tr))\leftor y.
\end{align*}
Hence, full sequential evaluation can be seen as a special case of short-circuit evaluation.
For example, it is a simple exercise to derive the SCL-translation of~\eqref{full} in $\SCLe$.

\bigskip

We finally note that a perhaps interesting variant of \FSCL\ 
is obtained by leaving out the constants \tr\ and \fa. 
Such a variant could be motivated by the fact that these constants are
usually absent in conditions in imperative programs.
However, in most programming languages the effect of
$\tr$ in a condition can be mimicked by a void equality 
test such as \texttt{(1==1)}, 
or in an expression-evaluated
programming language such as {Perl}, simply by the number \texttt{1}
(or any other non-zero number). 
In ``\FSCL\ without \tr\ and \fa'' the only 
\SCLe-axioms that remain are~\eqref{SCL2}, \eqref{SCL3}, and~\eqref{SCL7}, 
expressing duality and associativity.
Moreover, these axioms then yield a complete
axiomatisation of this restricted form of $se$-congruence.
Note that in this approach, connectives prescribing full sequential evaluation are
not definable, hence full sequential evaluation is not 
a special case of short-circuit evaluation.

However, we think that ``\SCL\ without \tr\ and \fa'' does not yield an appropriate 
point of view: in a sequential logic about truth and falsity 
one should be able to express the value \emph{true} itself.

\subsection{Side effects}
\label{subsec:effects}
Although side effects seem to be well understood in programming (see 
e.g., \cite{BW96,BW98,Nor97,KKS14}), they are often explained without a general
definition. In the following we consider side effects in
the context of the evaluation of propositional statements. 
The general question whether the sequential evaluation of a 
propositional statement has one or more side effects 
is context-dependent. 
Consider a toy programming language where assignments when evaluated as Boolean expressions
always yield \true\ and tests 
evaluate as expected. 
Some typical observations are these:
\begin{enumerate}\setlength\itemsep{-1mm}
\item
Consider the assignment $\texttt{(v:=5)}$ and observe its effect in the compound statements 
\[\texttt{(v:=5)} \texttt{ \&\& }\texttt{(v:=7)} \quad\text{and}\quad
\texttt{(v:=5)}\texttt{ \&\& } \texttt{(v==5)}.\]
In the first statement we cannot observe any side effect of the first assignment, i.e. 
changing it to assign a different value will never cause a different evaluation result, 
not even when the statement is embedded in a larger statement. 
We can say that the side effect of the first assignment is \emph{unobserved} in this context.

In the second compound statement however, changing the assigned value
will yield a different truth value for the compound statement and we can say that the side
effect of the assignment is \emph{observable} here. Note however that in a larger context
such as $\texttt{((}... \texttt{) \&\& (v:=6)) || (v:=6)}$ the side effect will again 
be unobserved.

\item
The side effect of the assignment \texttt{(v:=v+1)} is observable in a larger context,
as is that of \texttt{(v:=v-1)}. The side effect of the compound statement 
\texttt{(v:=v+1) \&\& (v:=v-1)} is however \emph{unobservable}, i.e., unobserved in all 
contexts. We can say that the side effects of these two assignments cancel out
provided these assignments occur adjacently.
\item
The question whether a test like $\texttt{(f(x)==5)}$ has an observable side effect
cannot be answered without examining the definition of the function \texttt{f}. 
Even if a programmer assumes that evaluating a call of \texttt{f} has one or more observable side 
effects, it is still possible to reason about the equivalence of compound statements 
containing this test.

\end{enumerate}
The above observations suggest that certain statements such as assignments and tests
are natural units for reasoning about side effects, and can be considered atomic when reasoning
about Boolean conditions as used in a programming language.
According to this view, $\FSCL$ preserves side effects of atoms in a very general sense
because it identifies only propositional statements with identical evaluation trees. 

The setting of short-circuit logic admits formal reasoning about side effects.
An example of such reasoning, 
building on observations 1 and 2 mentioned above,
is recorded in~\cite[Ex.7.2]{BP15}:
\begin{quote}
Assume 
atoms are of the form \texttt{($e$==$e'$)} and \texttt{(v:=$e$)}
with $\texttt v$ some  program variable and $e,e'$ arithmetical expressions over 
the integers that may contain $\texttt v$.
Furthermore, assume that \texttt{($e$==$e'$)} evaluates to \true\ if and only if
$e$ and $e'$
represent the same value, and \texttt{(v:=$e$)} always evaluates to 
\emph{true} with the effect that the value of $e$ is assigned to $\texttt v$.
Then all atoms satisfy the equation schemes~\eqref{rpaxioms}, and thus \RPSCL\ applies.\footnote{Of
  course, not all equations valid in this setting 
  follow from \RPSCL, e.g., $\RPSCL\not\vdash\texttt{(1==1)}=\tr$.
}
Note that the stronger identification $a\leftand a=a$ for all atoms $a$ is not valid: 
if $\texttt v$ has initial value 0 or 1, the statements
\[
(\texttt{(v:=v+1)}\leftand\texttt{(v:=v+1)})\leftand\texttt{(v==2)}
\quad\text{and}\quad
\texttt{(v:=v+1)}\leftand\texttt{(v==2)}
\]
evaluate to different results.
Finally, observe that for all initial values of \texttt{v} and for all $P\in\SP$,
\begin{equation}
\label{eq:perl}
\RPSCL\vdash \texttt{(v:=v+1)}\leftand (\texttt{(v:=v+1)}\leftor P)=
\texttt{(v:=v+1)}\leftand \texttt{(v:=v+1)},
\end{equation}
which agrees with the example in~\eqref{rp:ex}.
\end{quote}

We note that the set-up of our toy programming language suggests a sequential
variant of \emph{Dynamic Logic} (see, e.g.,~\cite{DynLog1})
in which assignments can be used both as tests and as programs.
Such a sequential variant could be appropriate for reasoning about side effects
and \RPSCL\ would be its underlying short-circuit logic.
However, if we assume that each assignment \texttt{(v:=$e$)}
evaluates to the Boolean value of $e$,
\RPSCL\ would not be the appropriate short-circuit logic.
For example, if in~\eqref{eq:perl} we take $P=\texttt{(v==0)}$ and set the initial value of 
\texttt{v} to $-2$, then
\[\texttt{(v:=v+1)}\leftand (\texttt{(v:=v+1)}\leftor \texttt{(v==0)})\]
evaluates to \true, while
$\texttt{(v:=v+1)}\leftand \texttt{(v:=v+1)}$ evaluates to \false. Hence, under
this interpretation of assignments, \FSCL\ is the appropriate short-circuit logic.

\section{Conclusions}
\label{sec:Conc}

In this paper we discussed \emph{free short-circuit logic} (\FSCL), following 
earlier research reported on in~\cite{BPS13,Daan,BP12a}.
In \FSCL, intermediate evaluation results are not memorised 
throughout the evaluation of a propositional statement, so evaluations of 
distinct occurrences of an atom may yield different truth values. 
The example on the condition a pedestrian evaluates before crossing a road with two-way traffic 
provides a clear motivation for this specific type of short-circuit evaluation.
The use of dedicated names and notation for connectives that \emph{prescribe} short-circuit 
evaluation is important in our approach
(in the area of computer science, one finds a wide variety of names and notations 
for short-circuit conjunction, such as ``logical and'' and ``conditional and'').
The symbols ${\leftand}$ and $\leftor$, as introduced in~\cite{BBR95} for four-valued logic
and named (left first) sequential conjunction and sequential 
disjunction, provide a convenient solution in this case.
Furthermore, this paper also contains an interesting exposition of the left-sequential connectives 
for the case of three-valued logics.

We note that defining the short-circuit evaluation function $se$ might have been avoided 
by defining the domain of the
short-circuit evaluation model $\Mse$ directly from the constants and connectives of \SigSCL.
We believe, however, that the function $se$ captures the operational nature of short-circuit evaluation
in a simple and elegant manner and should therefore be made explicit.

A last comment on the ten equational axioms that we selected for our axiomatisation of \FSCL\
(in~\cite{BPS13,Daan,BP12a}, a slightly different set of axioms is used).
Although evaluation trees provide an elegant way to model short-circuit evaluation 
in the presence of side effects, 
the equational axioms of \SCLe\ seem to grasp the nature of \FSCL-identities in a more 
direct way, and each of these axioms embodies a simple idea.
This is in particular the case for~\eqref{SCL1} and~\eqref{SCL3}, and that is why we kept 
these axioms in our definition of this axiom set (although they  can be derived from the remaining ones).

When it comes to reasoning about side effects, we subscribe to Parnas' view~\cite{Par10}:
\begin{quote}
\emph{Most mainline methods disparage side effects as a bad programming practice. 
Yet even in well-structured, reliable software, many components do have side 
effects; side effects are very useful in practice. It is time to investigate 
methods that deal with side effects as the normal case.}
\end{quote}
We hope that this paper establishes a step in this direction.

\bigskip

\emph{Future work.} Some specific questions are these: 
\begin{itemize}\setlength\itemsep{-1mm}
\item
Can Theorem~\ref{thm:sclcpl}, i.e.,
$\SCLe\vdash  s=t \iff\Mse\models s=t$ 
be generalised to open terms? 
(Note that $\Rightarrow$ is proved in Theorem~\ref{thm:sclsnd}.)
\item
Is there a simpler proof of Theorem~\ref{thm:sclcpl} than the one presented in this paper?
\item
Can the open problem from Section~\ref{subsec:4.2}, i.e.,
$\SCLe\vdash s=t \iff~ \FSCL\vdash s=t$,
be solved? (Note that Theorem~\ref{thm:openq} solves this for closed terms.)
\end{itemize}
Furthermore, we aim to 
provide elegant and independent equational axiomatisations for some other variants of SCL defined 
in~\cite{BP12a,BPS13}, or proofs of their non-existence when hidden operators are not involved.
And, last but not least, we aim to find fruitful applications for \FSCL\ and the other
SCLs we defined.

\paragraph*{Acknowledgement.}
We thank an anonymous reviewer and Inge Bethke for useful suggestions and discussion.

\appendix
\small
\section{Independence, normalisation, and \CP\ and evaluation trees}
\subsection{Independence proofs}
\label{App:scl}
We prove independence of the axioms of $\SCLi=\SCLe\setminus\{\eqref{SCL1},\eqref{SCL3}\}$.
All independence models that we use for this purpose were generated by \emph{Mace4}, 
see~\cite{BirdBrain}.
The independence of axiom~\eqref{SCL10} was shown in the proof of 
Theorem~\ref{thm:indepSCLe'}.
For all models \M\ defined below, we assume $\llbracket \fa\rrbracket^\M=0$ and 
$\llbracket \tr\rrbracket^\M=1$.
Furthermore, observe that all refutations below use at most one atom.

\medskip

A model \M\ for $\SCLi\setminus\{\eqref{SCL2}\}$
with domain $\{0,1\}$ that refutes $\fa \leftors\fa=\neg(\neg\tr\leftands\neg\tr)$ 
is this one:
\[
\begin{array}{r@{\hspace{5pt}}|@{\hspace{5pt}}c}
\neg\\\hline\\[-3mm]
0&1\\
1&1
\end{array}
\qquad\qquad
\begin{array}{r@{\hspace{5pt}}|@{\hspace{5pt}}c@{\hspace{5pt}}c@{\hspace{5pt}}c@{\hspace{5pt}}c@{\hspace{5pt}}c}
\leftands&
0&1\\\hline\\[-3mm]
0&
0&0\\
1&
0&1
\end{array}
\qquad\qquad
\begin{array}{r@{\hspace{5pt}}|@{\hspace{5pt}}c@{\hspace{5pt}}c@{\hspace{5pt}}c@{\hspace{5pt}}c@{\hspace{5pt}}c}
\leftors&
0&1\\\hline\\[-3mm]
0&
0&0\\
1&
1&0
\end{array}
\] 

A model \M\ for $\SCLi\setminus\{\eqref{SCL4}\}$
with domain $\{0,1\}$ that refutes $\tr \leftands\fa=\fa$ is the following:
\[
\begin{array}{r@{\hspace{5pt}}|@{\hspace{5pt}}c}
\neg\\\hline\\[-3mm]
0&0\\
1&1
\end{array}
\qquad\qquad
\begin{array}{r@{\hspace{5pt}}|@{\hspace{5pt}}c@{\hspace{5pt}}c@{\hspace{5pt}}c@{\hspace{5pt}}c@{\hspace{5pt}}c}
\leftands&
0&1\\\hline\\[-3mm]
0&
0&0\\
1&
1&1
\end{array}
\qquad\qquad
\begin{array}{r@{\hspace{5pt}}|@{\hspace{5pt}}c@{\hspace{5pt}}c@{\hspace{5pt}}c@{\hspace{5pt}}c@{\hspace{5pt}}c}
\leftors&
0&1\\\hline\\[-3mm]
0&
0&0\\
1&
1&1
\end{array}
\] 

A model \M\ for $\SCLi\setminus\{\eqref{SCL5}\}$ with domain $\{0,1\}$ that refutes 
$\tr \leftors\fa=\tr$ is this one:
\[
\begin{array}{r@{\hspace{5pt}}|@{\hspace{5pt}}c}
\neg\\\hline\\[-3mm]
0&0\\
1&0
\end{array}
\qquad\qquad
\begin{array}{r@{\hspace{5pt}}|@{\hspace{5pt}}c@{\hspace{5pt}}c@{\hspace{5pt}}c@{\hspace{5pt}}c@{\hspace{5pt}}c}
\leftands&
0&1\\\hline\\[-3mm]
0&
0&0\\
1&
0&1
\end{array}
\qquad\qquad
\begin{array}{r@{\hspace{5pt}}|@{\hspace{5pt}}c@{\hspace{5pt}}c@{\hspace{5pt}}c@{\hspace{5pt}}c@{\hspace{5pt}}c}
\leftors&
0&1\\\hline\\[-3mm]
0&
0&0\\
1&
0&0
\end{array}
\] 

A model \M\ for $\SCLi\setminus\{\eqref{SCL6}\}$ 
with domain $\{0,1,2\}$ and $\llbracket a\rrbracket^\M=2$
that refutes $\fa \leftands a=\fa$ is the following:
\[
\begin{array}{r@{\hspace{5pt}}|@{\hspace{5pt}}c}
\neg\\\hline\\[-3mm]
0&1\\
1&0\\
2&2
\end{array}
\qquad\qquad
\begin{array}{r@{\hspace{5pt}}|@{\hspace{5pt}}c@{\hspace{5pt}}c@{\hspace{5pt}}c@{\hspace{5pt}}c@{\hspace{5pt}}c}
\leftands&
0&1&2\\\hline\\[-3mm]
0&
0&0&2\\
1&
0&1&2\\
2&
0&2&2
\end{array}
\qquad\qquad
\begin{array}{r@{\hspace{5pt}}|@{\hspace{5pt}}c@{\hspace{5pt}}c@{\hspace{5pt}}c@{\hspace{5pt}}c@{\hspace{5pt}}c}
\leftors&
0&1&2\\\hline\\[-3mm]
0&
0&1&2\\
1&
1&1&2\\
2&
2&1&2
\end{array}
\] 

A model \M\ for $\SCLi\setminus\{\eqref{SCL7}\}$
with domain $\{0,1,2,3\}$ and $\llbracket a\rrbracket^\M=2$
that refutes $(a\leftands \fa) \leftands a = a\leftands (\fa \leftands a)$
is this one:
\[
\begin{array}{r@{\hspace{5pt}}|@{\hspace{5pt}}c}
\neg\\\hline\\[-3mm]
0&1\\
1&0\\
2&2\\
3&3
\end{array}
\qquad\qquad
\begin{array}{r@{\hspace{5pt}}|@{\hspace{5pt}}c@{\hspace{5pt}}c@{\hspace{5pt}}c@{\hspace{5pt}}c@{\hspace{5pt}}c}
\leftands&
0&1&2&3\\\hline\\[-3mm]
0&
0&0&0&0\\
1&
0&1&2&3\\
2&
3&2&0&3\\
3&
3&3&2&3
\end{array}
\qquad\qquad
\begin{array}{r@{\hspace{5pt}}|@{\hspace{5pt}}c@{\hspace{5pt}}c@{\hspace{5pt}}c@{\hspace{5pt}}c@{\hspace{5pt}}c}
\leftors&
0&1&2&3\\\hline\\[-3mm]
0&
0&1&2&3\\
1&
1&1&1&1\\
2&
2&3&1&3\\
3&
3&3&2&3
\end{array}
\] 

A model \M\ for $\SCLi\setminus\{\eqref{SCL8}\}$
with domain $\{0,1,2,3\}$ and $\llbracket a\rrbracket^\M=2$
that refutes $\neg a\leftands\fa = a\leftands\fa$ 
is the following:
\[
\begin{array}{r@{\hspace{5pt}}|@{\hspace{5pt}}c}
\neg\\\hline\\[-3mm]
0&1\\
1&0\\
2&3\\
3&2
\end{array}
\qquad\qquad
\begin{array}{r@{\hspace{5pt}}|@{\hspace{5pt}}c@{\hspace{5pt}}c@{\hspace{5pt}}c@{\hspace{5pt}}c@{\hspace{5pt}}c}
\leftands&
0&1&2&3\\\hline\\[-3mm]
0&
0&0&0&0\\
1&
0&1&2&3\\
2&
2&2&2&2\\
3&
3&3&3&3
\end{array}
\qquad\qquad
\begin{array}{r@{\hspace{5pt}}|@{\hspace{5pt}}c@{\hspace{5pt}}c@{\hspace{5pt}}c@{\hspace{5pt}}c@{\hspace{5pt}}c}
\leftors&
0&1&2&3\\\hline\\[-3mm]
0&
0&1&2&3\\
1&
1&1&1&1\\
2&
2&2&2&2\\
3&
3&3&3&3
\end{array}
\] 

A model \M\ for $\SCLi\setminus\{\eqref{SCL9}\}$ 
with domain $\{0,1,2,3,4\}$ and $\llbracket a\rrbracket^\M=2$
that refutes $(a \leftands \fa) \leftors a =  (a \leftors \tr)\leftands a$
is this one:
\[
\begin{array}{r@{\hspace{5pt}}|@{\hspace{5pt}}c}
\neg\\\hline\\[-3mm]
0&1\\
1&0\\
2&2\\
3&4\\
4&3
\end{array}
\qquad\qquad
\begin{array}{r@{\hspace{5pt}}|@{\hspace{5pt}}c@{\hspace{5pt}}c@{\hspace{5pt}}c@{\hspace{5pt}}c@{\hspace{5pt}}c}
\leftands&
0&1&2&3&4\\\hline\\[-3mm]
0&
0&0&0&0&0\\
1&
0&1&2&3&4\\
2&
3&2&2&3&2\\
3&
3&3&3&3&3\\
4&
3&4&4&3&4
\end{array}
\qquad\qquad
\begin{array}{r@{\hspace{5pt}}|@{\hspace{5pt}}c@{\hspace{5pt}}c@{\hspace{5pt}}c@{\hspace{5pt}}c@{\hspace{5pt}}c}
\leftors&
0&1&2&3&4\\\hline\\[-3mm]
0&
0&1&2&3&4\\
1&
1&1&1&1&1\\
2&
2&4&2&2&4\\
3&
3&4&3&3&4\\
4&
4&4&4&4&4
\end{array}
\] 

\subsection{Correctness of the normalisation function}
\label{app:nf}
In order to prove that $\nfs: \ST \to \SNF$ is indeed a normalisation function
we need to prove that for all $\SCL$-terms $P$, $\nfs(P)$ terminates, $\nfs(P)
\in \SNF$ and $\EqFSCL \vdash \nfs(P) = P$. To arrive at this result, we prove
several intermediate results about the functions $\nfs^n$ and $\nfs^c$ in the
order in which their definitions were presented in Section~\ref{subsec:snf}. For
the sake of brevity we will not explicitly prove that these functions
terminate. To see that each function terminates consider that a termination
proof would closely mimic the proof structure of the lemmas dealing with the
grammatical categories of the images of these functions.

\begin{lemma}
\label{lem:ptpfs}
For all $P\in P^\fa$ and $Q\in P^\tr$, $\EqFSCL \vdash P = P \leftands
x$ and $\EqFSCL \vdash Q = Q \leftors x$.
\end{lemma}

\begin{proof}
We prove both claims simultaneously by induction. In the base case we have
$\fa = \fa \leftands x$ by axiom~\eqref{SCL6}. The base case for the
second claim follows from that for the first claim by duality.

For the induction we have $(a \leftors P_1) \leftands P_2 = (a
\leftors P_1) \leftands (P_2 \leftands x)$ by the induction
hypothesis and the result follows from \eqref{SCL7}. For the second claim we
again appeal to duality.
\end{proof}

The equality we showed as an example in Lemma~\ref{lem:seqs} will prove
useful in this appendix, as will the following equalities, which also deal
with terms of the form $x \leftands \fa$ and $x \leftors \tr$.
In the sequel, we refer to the dual of axiom $(n)$ by $(n)'$.

\begin{lemma}
\label{lem:seqs2}
The following equations can all be derived from $\EqFSCL$.
\textup{
\begin{enumerate}\setlength\itemsep{-1mm}
\item
$(x \leftors (y \leftands \fa)) \leftands (z \leftands \fa) = (\neg
  x \leftors (z \leftands \fa)) \leftands (y \leftands \fa),$
  \label{eq:b2}
\item
$(x \leftands (y \leftors \tr)) \leftors (z \leftands \fa) 
= (x
  \leftors (z \leftands \fa)) \leftands (y \leftors \tr),$
  \label{eq:b3}
\item
$(x \leftors \tr) \leftands \neg y = \neg((x \leftors \tr) \leftands
  y),$
  \label{eq:b4}
\item
$(x \leftands (y \leftands (z \leftors \tr))) \leftors
  (w \leftands (z \leftors \tr)) = ((x \leftands y) \leftors w) \leftands
  (z \leftors \tr),$
  \label{eq:b5}
\item
$(x \leftors ((y \leftors \tr) \leftands (z \leftands \fa)))
  \leftands ((w \leftors \tr) \leftands (z \leftands \fa)) =
  ((x \leftands (w \leftors \tr)) \leftors (y \leftors \tr)) \leftands (z
  \leftands \fa),$
  \label{eq:b6}
\item
$(x \leftors ((y \leftors \tr) \leftands (z \leftands \fa)))
  \leftands (w \leftands \fa) = ((\neg x \leftands (y \leftors \tr))
  \leftors (w \leftands \fa)) \leftands (z \leftands \fa).$
  \label{eq:b7}
\end{enumerate}
}
\end{lemma}

\begin{proof}
See Table~\ref{tab:proof}.
We note that these equations were checked with the theorem prover \emph{Prover9}~\cite{BirdBrain}.
\end{proof}

\begin{table}
{ \small
\centering
\rule{1\textwidth}{.4pt}
\begin{align*}
\eqref{eq:b2}\quad
(x&\leftors (y \leftands \fa)) \leftands (z \leftands
  \fa)\\
&= (\neg x \leftors (z \leftands \fa)) \leftands ((y \leftands \fa)
  \leftands (z \leftands \fa))
&&\text{by Lemma~\ref{lem:seqs}} \\
&= (\neg x \leftors (z \leftands \fa)) \leftands (y \leftands \fa),
&&\text{by \eqref{SCL6}, \eqref{SCL7}} 
\\[2mm]
\eqref{eq:b3}\quad
(x&\leftands (y \leftors \tr)) \leftors (z \leftands
  \fa)\\
&= (x \leftors (z \leftands \fa)) \leftands ((y \leftors \tr) \leftors (z
  \leftands \fa))
&&\text{by \eqref{SCL10}} \\
&= (x \leftors (z \leftands \fa)) \leftands (y \leftors \tr),
&&\text{by $\eqref{SCL6}'$, $\eqref{SCL7}'$}
\\[2mm]
\eqref{eq:b4}\quad
(x &\leftors \tr)\leftands \neg y\\
&= \neg((\neg x \leftands \fa) \leftors y)
&&\text{by \eqref{SCL2}} \\
&= \neg((x \leftands \fa) \leftors y)
&&\text{by \eqref{SCL8}} \\
&= \neg((x \leftors \tr) \leftands y),
&&\text{by \eqref{SCL9}}
\\[2mm]
\eqref{eq:b5}\quad
(x &\leftands (y \leftands (z \leftors \tr))) \leftors
  (w \leftands (z \leftors \tr)) \\
&= ((x \leftands y) \leftands (z \leftors \tr)) \leftors
  (w \leftands (z \leftors \tr))
&&\text{by \eqref{SCL7}} \\
&= ((x \leftands y) \leftors w) \leftands (z \leftors \tr),
&&\text{by $\eqref{SCL10}'$}
\\[2mm]
\eqref{eq:b6}\quad
(x &\leftors ((y \leftors \tr) \leftands (z \leftands \fa)))
  \leftands ((w \leftors \tr) \leftands (z \leftands \fa)) \\
&= (x \leftors ((y \leftands \fa) \leftors (z \leftands \fa)))
\leftands~((w \leftors \tr) \leftands (z \leftands \fa))
&&\text{by \eqref{SCL9}} \\
&= ((x \leftors (y \leftands \fa)) \leftors (z \leftands \fa))
\leftands((w \leftors \tr) \leftands (z \leftands \fa))
&&\text{by $\eqref{SCL7}'$} \\
&= (\neg(x \leftors (y \leftands \fa)) \leftors (w \leftors \tr))
  \leftands (z \leftands \fa)
&&\text{by Lemma~\ref{lem:seqs}} \\
&= ((\neg x \leftands (\neg y \leftors \tr)) \leftors (w \leftors \tr))
  \leftands (z \leftands \fa)
&&\text{by \eqref{SCL2}$'$} \\
&= ((\neg x \leftands (y \leftors \tr)) \leftors (w \leftors \tr))
  \leftands (z \leftands \fa)
&&\text{by $\eqref{SCL8}'$} \\
&= ((\neg x \leftands (y \leftors \tr)) \leftors (w \leftors (\tr
  \leftors (y \leftors \tr))))\leftands (z \leftands \fa)
&&\text{by $\eqref{SCL6}'$} \\
&= ((\neg x \leftands (y \leftors \tr)) \leftors ((w \leftors \tr)
  \leftors (y \leftors \tr)))\leftands(z \leftands \fa)
&&\text{by $\eqref{SCL7}'$} \\
&= ((x \leftands (w \leftors \tr)) \leftors (y \leftors \tr))
  \leftands (z \leftands \fa),
&&\text{by Lemma~$\ref{lem:seqs}'$} 
\\[2mm]
\eqref{eq:b7}\quad
(x&\leftors ((y \leftors \tr) \leftands (z \leftands
  \fa))) \leftands (w \leftands \fa) \\
&= (\neg x \leftors (w \leftands \fa)) \leftands (((y \leftors \tr)
  \leftands (z \leftands \fa))\leftands~ (w \leftands \fa))
&&\text{by Lemma~\ref{lem:seqs}} \\
&= (\neg x \leftors (w \leftands \fa)) \leftands ((y \leftors \tr)
  \leftands (z \leftands \fa))
&&\text{by \eqref{SCL6}, \eqref{SCL7}} \\
&= ((\neg x \leftors (w \leftands \fa)) \leftands (y \leftors \tr))
  \leftands (z \leftands \fa)
&&\text{by \eqref{SCL7}} \\
&= ((\neg x \leftands (y \leftors \tr)) \leftors ((w \leftands \fa)
  \leftands (y \leftors \tr)))\leftands(z \leftands \fa)
&&\text{by $\eqref{SCL10}'$} \\
&= ((\neg x \leftands (y \leftors \tr)) \leftors (w \leftands \fa))
  \leftands (z \leftands \fa).
&&\text{by \eqref{SCL6}, \eqref{SCL7}}
  &
\end{align*}
\hrule
}
\caption{Derivations for the proof of Lemma~\ref{lem:seqs2}}
\label{tab:proof}
\end{table}

\begin{lemma}
\label{lem:nfsn}
For all $P \in \SNF$, if $P$ is a $\tr$-term then $\nfs^n(P)$ is an
$\fa$-term, if it is an $\fa$-term then $\nfs^n(P)$ is a $\tr$-term, if
it is a $\tr$-$*$-term then so is $\nfs^n(P)$, and
\begin{equation*}
\EqFSCL \vdash \nfs^n(P) = \neg P.
\end{equation*}
\end{lemma}
\begin{proof}
We first prove the claims for $\tr$-terms, by induction on $P^\tr$.  In the
base case $\nfs^n(\tr) = \fa$ by \eqref{eq:nfsn1}, so
$\nfs^n(\tr)$ is an
$\fa$-term. The claim that $\EqFSCL \vdash \nfs^n(\tr) = \neg \tr$ is
immediate by \eqref{SCL1}. For the inductive case we have that $\nfs^n((a
\leftands P^\tr) \leftors Q^\tr) = (a \leftors \nfs^n(Q^\tr)) \leftands
\nfs^n(P^\tr)$ by \eqref{eq:nfsn2}, where we assume that $\nfs^n(P^\tr)$ and 
$\nfs^n(Q^\tr)$
are $\fa$-terms and that $\EqFSCL \vdash \nfs^n(P^\tr) = \neg P^\tr$
and $\EqFSCL \vdash \nfs^n(Q^\tr) = \neg Q^\tr$. It follows from the
induction hypothesis that $\nfs^n((a \leftands P^\tr) \leftors Q^\tr)$ is
an $\fa$-term. Furthermore, noting that by the induction hypothesis we may
assume that $\nfs^n(P^\tr)$ and $\nfs^n(Q^\tr)$ are $\fa$-terms, we
have:
\begin{align*}
\nfs^n((a\leftands P^\tr) \leftors Q^\tr)
&= (a \leftors \nfs^n(Q^\tr)) \leftands \nfs^n(P^\tr)
&&\text{by \eqref{eq:nfsn2}} \\
&= (a \leftors (\nfs^n(Q^\tr) \leftands \fa)) \leftands (\nfs^n(P^\tr)
  \leftands \fa)
&&\text{by Lemma~\ref{lem:ptpfs}} \\
&= (\neg a \leftors (\nfs^n(P^\tr) \leftands \fa)) \leftands
  (\nfs^n(Q^\tr) \leftands \fa )
&&\text{by Lemma~\ref{lem:seqs2}.\ref{eq:b2}} \\
&= (\neg a \leftors \nfs^n(P^\tr)) \leftands \nfs^n(Q^\tr)
&&\text{by Lemma~\ref{lem:ptpfs}} \\
&= (\neg a \leftors \neg P^\tr) \leftands \neg Q^\tr
&&\text{by induction hypothesis} \\
&= \neg((a \leftands P^\tr) \leftors Q^\tr).
&&\text{by \eqref{SCL2} and its dual}
\end{align*}

For $\fa$-terms we prove our claims by induction on $P^\fa$. In the base
case $\nfs^n(\fa) = \tr$ by \eqref{eq:nfsn3}, so $\nfs^n(\fa)$ is a
$\tr$-term. The claim that $\EqFSCL \vdash \nfs^n(\fa) = \neg \fa$ is
immediate by~\eqref{SCL1}$'$. For the inductive case we have that
$\nfs^n((a \leftors P^\fa) \leftands Q^\fa) = (a \leftands
\nfs^n(Q^\fa)) \leftors \nfs^n(P^\fa)$ by \eqref{eq:nfsn4}, where we assume that
$\nfs^n(P^\fa)$ and $\nfs^n(Q^\fa)$ are $\tr$-terms and that $\EqFSCL
\vdash \nfs^n(P^\fa) = \neg P^\fa$ and $\EqFSCL \vdash \nfs^n(Q^\fa) =
\neg Q^\fa$. It follows from the induction hypothesis that $\nfs^n((a
\leftors P^\fa) \leftands Q^\fa)$ is a $\tr$-term. Furthermore, noting
that by the induction hypothesis we may assume that $\nfs^n(P^\fa)$ and
$\nfs^n(Q^\fa)$ are $\tr$-terms, the proof of derivably equality is dual
to that for $\nfs^n((a \leftands P^\tr) \leftors Q^\tr)$.

To prove the lemma for $\tr$-$*$-terms we first verify that the auxiliary
function $\nfs^n_1$ returns a $*$-term and that for any $*$-term $P$, $\EqFSCL
\vdash \nfs^n_1(P) = \neg P$. We show this by induction on the number of
$\ell$-terms in $P$. For the base cases it is immediate by the above cases
for $\tr$-terms and $\fa$-terms that $\nfs^n_1(P)$ is a $*$-term.
Furthermore, if $P$ is an $\ell$-term of the form $(a\leftands P^\tr) \leftors Q^\fa$ we
have:
\begin{align*}
\nfs^n_1((a\leftands P^\tr) \leftors Q^\fa)
&= (\neg a \leftands \nfs^n(Q^\fa)) \leftors \nfs^n(P^\tr)
&&\text{by \eqref{eq:nfsn6}} \\
&= (\neg a \leftands (\nfs^n(Q^\fa) \leftors \tr)) \leftors
  (\nfs^n(P^\tr) \leftands \fa)
&&\text{by Lemma~\ref{lem:ptpfs}} \\
&= (\neg a \leftors (\nfs^n(P^\tr) \leftands \fa)) \leftands
  (\nfs^n(Q^\fa) \leftors \tr)
&&\text{by Lemma~\ref{lem:seqs2}.\ref{eq:b3}} \\
&= (\neg a \leftors \nfs^n(P^\tr)) \leftands \nfs^n(Q^\fa)
&&\text{by Lemma~\ref{lem:ptpfs}} \\
&= (\neg a \leftors \neg P^\tr) \leftands \neg Q^\fa
&&\text{by induction hypothesis} \\
&= \neg((a \leftands P^\tr) \leftors Q^\fa).
&&\text{by \eqref{SCL2} and its dual}
\end{align*}
If $P$ is an $\ell$-term of the form $(\neg a\leftands P^\tr) \leftors Q^\fa$ the proof proceeds
the same, substituting $\neg a$ for $a$ and applying \eqref{eq:nfsn7}
and \eqref{SCL3} where
needed. For the inductive step we assume that the result holds for all
$*$-terms with fewer $\ell$-terms than $P^* \leftands Q^d$ and $P^* \leftors
Q^c$. By \eqref{eq:nfsn8} and \eqref{eq:nfsn9},
each application of $\nfs^n_1$ changes the main connective
(not occurring inside an $\ell$-term) and hence the result is a $*$-term.
Derivable equality is, given the induction hypothesis, an instance of~\eqref{SCL2}$'$.

With this result we can now see that $\nfs^n(P^\tr \leftands Q^*)$ is indeed
a $\tr$-$*$-term. We note that, by the above, Lemma~\ref{lem:ptpfs}
implies that $\neg P^\tr = \neg P^\tr \leftands \fa$. Now we find that:
\begin{align*}
\nfs^n(P^\tr \leftands Q^*)
&= P^\tr \leftands \nfs^n_1(Q^*)
&&\text{by \eqref{eq:nfsn5}} \\
&= P^\tr \leftands \neg Q^*
&&\text{as shown above} \\
&= (P^\tr \leftors \tr) \leftands \neg Q^*
&&\text{by Lemma~\ref{lem:ptpfs}} \\
&= \neg((P^\tr \leftors \tr) \leftands Q^*)
&&\text{by Lemma~\ref{lem:seqs2}.\ref{eq:b4}} \\
&= \neg(P^\tr \leftands Q^*).
&&\text{by Lemma~\ref{lem:ptpfs}}
\end{align*}
Hence for all $P \in \SNF$, $\EqFSCL \vdash \nfs^n(P) = \neg P$.
\end{proof}

\begin{lemma}
\label{lem:nfsc1}
For any $\tr$-term $P$ and $Q \in \SNF$, $\nfs^c(P, Q)$ has the same
grammatical category as $Q$ and
\begin{equation*}
\EqFSCL \vdash \nfs^c(P, Q) = P \leftands Q.
\end{equation*}
\end{lemma}
\begin{proof}
By induction on the complexity of the $\tr$-term. In the base case we see
that $\nfs^c(\tr, P) = P$ by \eqref{eq:nfsc1}, which is clearly of the same grammatical category
as $P$. Derivable equality is an instance of \eqref{SCL4}.

For the inductive step we assume that the result holds for all $\tr$-terms of
lesser complexity than $(a \leftands P^\tr)\leftors Q^\tr$. The claim about the grammatical
category follows immediately from the induction hypothesis. For the claim about
derivable equality we make a case distinction on the grammatical category of
the second argument. If the second argument is a $\tr$-term, we prove
derivable equality as follows:
\begin{align*}
\nfs^c((a&\leftands P^\tr) \leftors Q^\tr, R^\tr) \\
&= (a \leftands \nfs^c(P^\tr, R^\tr)) \leftors \nfs^c(Q^\tr, R^\tr)
&&\text{by \eqref{eq:nfsc2}} \\
&= (a \leftands (P^\tr \leftands R^\tr)) \leftors (Q^\tr \leftands
  R^\tr)
&&\text{by induction hypothesis} \\
&= (a \leftands (P^\tr \leftands (R^\tr \leftors \tr))) \leftors (Q^\tr
  \leftands (R^\tr \leftors \tr))
&&\text{by Lemma~\ref{lem:ptpfs}} \\
&= ((a \leftands P^\tr) \leftors Q^\tr) \leftands (R^\tr \leftors \tr)
&&\text{by Lemma~\ref{lem:seqs2}.\ref{eq:b5}} \\
&= ((a \leftands P^\tr) \leftors Q^\tr) \leftands R^\tr.
&&\text{by Lemma~\ref{lem:ptpfs}}
\end{align*}
If the second argument is an $\fa$-term, we prove derivable equality as
follows:
\begin{align*}
\nfs^c((a &\leftands P^\tr) \leftors Q^\tr, R^\fa) \\
&= (a \leftors \nfs^c(Q^\tr, R^\fa)) \leftands \nfs^c(P^\tr, R^\fa)
&&\text{by \eqref{eq:nfsc3}} \\
&= (a \leftors (Q^\tr \leftands R^\fa)) \leftands (P^\tr \leftands
  R^\fa)
&&\text{by induction hypothesis} \\
&= (a \leftors ((Q^\tr \leftors \tr) \leftands (R^\fa \leftands \fa)))
  \leftands~\\
  &\phantom{~=}((P^\tr \leftors \tr) \leftands (R^\fa \leftands \fa))
&&\text{by Lemma~\ref{lem:ptpfs}} \\
&= ((a \leftands (P^\tr \leftors \tr)) \leftors (Q^\tr \leftors \tr))
  \leftands (R^\fa \leftands \fa)
&&\text{by Lemma~\ref{lem:seqs2}.\ref{eq:b6}} \\
&= ((a \leftands P^\tr) \leftors Q^\tr) \leftands R^\fa.
&&\text{by Lemma~\ref{lem:ptpfs}}
\end{align*}

If the second argument is $\tr$-$*$-term, the result follows 
by \eqref{eq:nfsc4} from the case
where the second argument is a $\tr$-term, and \eqref{SCL7}.
\end{proof}

\begin{lemma}
\label{lem:nfsc2}
For any $\fa$-term $P$ and $Q \in \SNF$, $\nfs^c(P, Q)$ is an
$\fa$-term and
\begin{equation*}
\EqFSCL \vdash \nfs^c(P, Q) = P \leftands Q.
\end{equation*}
\end{lemma}
\begin{proof}
The grammatical result is immediate by \eqref{eq:nfsc5}
and the claim about derivable equality
follows from Lemma~\ref{lem:ptpfs}, \eqref{SCL7} and \eqref{SCL6}.
\end{proof}

\begin{lemma}
\label{lem:nfsc3}
For any $\tr$-$*$-term $P$ and $\tr$-term $Q$, $\nfs^c(P, Q)$ has the same
grammatical category as $P$ and
\begin{equation*}
\EqFSCL \vdash \nfs^c(P, Q) = P \leftands Q.
\end{equation*}
\end{lemma}
\begin{proof}
By \eqref{eq:nfsc6} and \eqref{SCL7}
it suffices to prove the claims for $\nfs^c_1$, i.e., that
$\nfs^c_1(P^*, Q^\tr)$ is a $*$-term and that $\EqFSCL \vdash \nfs^c_1(P^*,
Q^\tr) = P^* \leftands Q^\tr$. We prove this by induction on the number of
$\ell$-terms in $P^*$. In the base case we deal with $\ell$-terms and the
grammatical claim follows from Lemma~\ref{lem:nfsc1}. We prove derivable
equality as follows, letting $\hat{a} \in \{a, \neg a\}$:
\begin{align*}
\nfs^c_1((\hat{a} \leftands P^\tr) \leftors Q^\fa, R^\tr)
&= (\hat{a} \leftands \nfs^c(P^\tr, R^\tr)) \leftors Q^\fa
&&\text{by \eqref{eq:nfsc7}, \eqref{eq:nfsc8}} \\
&= (\hat{a} \leftands (P^\tr \leftands R^\tr)) \leftors Q^\fa
&&\text{by Lemma~\ref{lem:nfsc1}} \\
&= ((\hat{a} \leftands P^\tr) \leftands R^\tr) \leftors Q^\fa
&&\text{by \eqref{SCL7}} \\
&= ((\hat{a} \leftands P^\tr) \leftands (R^\tr \leftors \tr))
  \leftors (Q^\fa \leftands \fa)
&&\text{by Lemma~\ref{lem:ptpfs}} \\
&= ((\hat{a} \leftands P^\tr) \leftors (Q^\fa \leftands \fa))
  \leftands (R^\tr \leftors \tr)
&&\text{by Lemma~\ref{lem:seqs2}.\ref{eq:b3}} \\
&= ((\hat{a} \leftands P^\tr) \leftors Q^\fa) \leftands R^\tr.
&&\text{by Lemma~\ref{lem:ptpfs}}
\end{align*}

For the induction step we assume that the result holds for all $*$-terms with
fewer $\ell$-terms than $P^* \leftands Q^d$ and $P^* \leftors Q^c$.  In the
case of conjunctions the results follow from \eqref{eq:nfsc9}, the induction 
hypothesis, and \eqref{SCL7}. 
In the case of disjunctions the results follow immediately
from \eqref{eq:nfsc10}, the induction hypothesis, Lemma~\ref{lem:ptpfs}, 
and \eqref{SCL10}$'$.
\end{proof}

\begin{lemma}
\label{lem:nfsc4}
For any $\tr$-$*$-term $P$ and $\fa$-term $Q$, $\nfs^c(P, Q)$ is an
$\fa$-term and
\begin{equation*}
\EqFSCL \vdash \nfs^c(P, Q) = P \leftands Q.
\end{equation*}
\end{lemma}
\begin{proof}
By \eqref{eq:nfsc11}, Lemma~\ref{lem:nfsc1} and \eqref{SCL7} it suffices 
to prove that
$\nfs^c_2(P^*, Q^\fa)$ is an $\fa$-term and that $\EqFSCL \vdash
\nfs^c_2(P^*, Q^\fa) = P^* \leftands Q^\fa$. We prove this by induction
on the number of $\ell$-terms in $P^*$. In the base case we deal with
$\ell$-terms and the grammatical claim follows from Lemma~\ref{lem:nfsc1}. We
derive the remaining claim for $\ell$-terms of the form $(a\leftands P^\tr) \leftors Q^\fa$
as:
\begin{align*}
\nfs^c_2((a \leftands P^\tr) \leftors Q^\fa, R^\fa)
&= (a \leftors Q^\fa) \leftands \nfs^c(P^\tr, R^\fa)
&&\text{by \eqref{eq:nfsc12}} \\
&= (a \leftors Q^\fa) \leftands (P^\tr \leftands R^\fa)
&&\text{by Lemma~\ref{lem:nfsc1}} \\
&= ((a \leftors Q^\fa) \leftands P^\tr) \leftands R^\fa
&&\text{by \eqref{SCL7}} \\
&= ((a \leftors (Q^\fa \leftands \fa)) \leftands (P^\tr \leftors
  \tr)) \leftands R^\fa
&&\text{by Lemma~\ref{lem:ptpfs}} \\
&= ((a \leftands (P^\tr \leftors \tr)) \leftors (Q^\fa \leftands
  \fa)) \leftands R^\fa
&&\text{by Lemma~\ref{lem:seqs2}.\ref{eq:b3}} \\
&= ((a \leftands P^\tr) \leftors Q^\fa) \leftands R^\fa.
&&\text{by Lemma~\ref{lem:ptpfs}}
\end{align*}
For $\ell$-terms of the form $(\neg a\leftands P^\tr) \leftors Q^\fa$ we derive:
\begin{align*}
\nfs^c_2((&\neg a\leftands P^\tr) \leftors Q^\fa, R^\fa) \\
&= (a \leftors \nfs^c(P^\tr, R^\fa)) \leftands Q^\fa
&&\text{by \eqref{eq:nfsc13}} \\
&= (a \leftors (P^\tr \leftands R^\fa)) \leftands Q^\fa
&&\text{by induction hypothesis} \\
&= (a \leftors ((P^\tr \leftors \tr) \leftands (R^\fa \leftands
  \fa))) \leftands (Q^\fa \leftands \fa)
&&\text{by Lemma~\ref{lem:ptpfs}} \\
&= ((\neg a \leftands (P^\tr \leftors \tr)) \leftors (Q^\fa \leftands
  \fa)) \leftands (R^\fa \leftands \fa)
&&\text{by Lemma~\ref{lem:seqs2}.\ref{eq:b7}} \\
&= ((\neg a \leftands P^\tr) \leftors Q^\fa) \leftands R^\fa.
&&\text{by Lemma~\ref{lem:ptpfs}}
\end{align*}

For the induction step we assume that the result holds for all $*$-terms with
fewer $\ell$-terms than $P^* \leftands Q^d$ and $P^* \leftors Q^c$.  In the
case of conjunctions the results follow from \eqref{eq:nfsc14},
the induction hypothesis, and
\eqref{SCL7}. In the case of disjunctions note that by Lemma~\ref{lem:nfsn}
and the proof of Lemma~\ref{lem:nfsc3}, we have that $\nfs^n(\nfs^c_1(P^*,
\nfs^n(R^\fa)))$ is a $*$-term with same number of $\ell$-terms as $P^*$.
The grammatical result follows from this fact, \eqref{eq:nfsc15}, and the 
induction hypothesis.
Furthermore, noting that by the same argument $\nfs^n(\nfs^c_1(P^*,
\nfs^n(R^\fa))) = \neg(P^* \leftands \neg R^\fa)$, we derive:
\begin{align*}
\nfs^c_2(P^* \leftors Q^c, R^\fa)
&= \nfs^c_2(\nfs^n(\nfs^c_1(P^*, \nfs^n(R^\fa))), \nfs^c_2(Q^c, R^\fa))
&&\text{by \eqref{eq:nfsc15}} \\
&= \nfs^n(\nfs^c_1(P^*, \nfs^n(R^\fa))) \leftands (Q^c \leftands R^\fa)
&&\text{by induction hypothesis} \\
&= \neg(P^* \leftands \neg R^\fa) \leftands (Q^c \leftands R^\fa)
&&\text{as shown above} \\
&= (\neg P^* \leftors R^\fa) \leftands (Q^c \leftands R^\fa)
&&\text{by \eqref{SCL3}, \eqref{SCL2}} \\
&= (\neg P^* \leftors (R^\fa \leftands \fa)) \leftands (Q^c \leftands
  (R^\fa \leftands \fa))
&&\text{by Lemma~\ref{lem:ptpfs}} \\
&= (P^* \leftors Q^c) \leftands (R^\fa \leftands \fa)
&&\text{by Lemma~\ref{lem:seqs}} \\
&= (P^* \leftors Q^c) \leftands R^\fa.
&&\text{by Lemma~\ref{lem:ptpfs}}
\end{align*}
This completes the proof.
\end{proof}

\begin{lemma}
\label{lem:nfsc5}
For any $P, Q \in \SNF$, $\nfs^c(P, Q)$ is in $\SNF$ and
$\EqFSCL \vdash \nfs^c(P, Q) = P \leftands Q$.
\end{lemma}

\begin{proof}
By the four preceding lemmas it suffices to show that 
\[\nfs^c(P^\tr \leftands Q^*, R^\tr \leftands S^*)\]
is in $\SNF$ and that $\EqFSCL \vdash
\nfs^c(P^\tr \leftands Q^*, R^\tr \leftands S^*) = (P^\tr \leftands Q^*)
\leftands (R^\tr \leftands S^*)$. By \eqref{SCL7} and \eqref{eq:nfsc16}, in turn, it suffices to
prove that $\nfs^c_3(P^*, Q^\tr \leftands R^*)$ is a $*$-term and that
$\EqFSCL \vdash \nfs^c_3(P^*, Q^\tr \leftands R^*) = P^* \leftands (Q^\tr
\leftands R^*)$. We prove this by induction on the number of $\ell$-terms in
$R^*$. In the base case we have that $\nfs^c_3(P^*, Q^\tr \leftands R^\ell)
= \nfs^c_1(P^*, Q^\tr) \leftands R^\ell$ by \eqref{eq:nfsc17} and
the lemma's statement follows from Lemma~\ref{lem:nfsc3} and \eqref{SCL7}.

For conjunctions the lemma's statement follows from the induction hypothesis,
\eqref{SCL7} and \eqref{eq:nfsc18},
and for disjunctions it follows from Lemma~\ref{lem:nfsc3},
 \eqref{SCL7} and \eqref{eq:nfsc19}.
\end{proof}

We can now easily prove Theorem~\ref{thm:nfs}:
\begin{thm:nfs*}[Normal forms]
For any $P \in \ST$, $\nfs(P)$ terminates, $\nfs(P) \in \SNF$ and 
\[\EqFSCL\vdash \nfs(P) = P.\]
\end{thm:nfs*}

\begin{proof}
By induction on the structure of $P$. If $P$ is an atom, the result follows from
\eqref{eq:nfs1} and axioms~\eqref{SCL4}, \eqref{SCL5} and its dual. 
If $P$ is $\tr$ or $\fa$
the result follows from by \eqref{eq:nfs2} or \eqref{eq:nfs3}. 
For the induction we get the result from definitions \eqref{eq:nfs4}-\eqref{eq:nfs6}, Lemma
\ref{lem:nfsn}, Lemma~\ref{lem:nfsc5}, and axiom~\eqref{SCL2}.
\end{proof}

\subsection{\CP\ and evaluation trees}
\label{app:cp}
We finally show that equivalence~\eqref{**} holds: this is Theorem~\ref{thm:A1} below and
this text (excluding footnotes) is taken from~\cite{BP15}.

\bigskip

Let $\PS$ be the set of closed terms over $\SigCP$, and recall $se$'s definition on $\PS$
from Section~\ref{subsec:4.1}. 
\begin{definition}
\label{def:basic}
\textbf{Basic forms over $A$} are defined by the following grammar
\[t::= \tr\mid\fa\mid t\lef a \rig t\quad\text{for $a\in A$.}\]
We write $\BF$ for the set of basic forms over $A$. 
\end{definition}

The following lemma's exploit the structure of basic forms.\footnote{%
  We speak of ``basic forms'' instead of normal forms in order to avoid intuitions 
  from term rewriting: for example, the basic form associated with action $a$
  is $\tr\lef a\rig\fa$, whereas one would expect that the normal form of the latter is $a$.}

\begin{lemma}
\label{la:2.5}
For each $P\in\PS$ there exists $Q\in\BF$ such that
$\CP\vdash P=Q$.
\end{lemma}
\begin{proof} First we establish an auxiliary result:
 if $P,Q,R$ are basic forms, then there is a basic form $S$ such that
$\CP\vdash P\lef Q\rig R=S$. This follows by structural induction on $Q$.

The lemma's statement follows by structural induction on $P$. The base
cases $P\in\{\tr,\fa,a\mid a\in A\}$ are trivial, and if $P= P_1\lef P_2\rig
P_3$ there exist by induction basic forms $Q_i$ such that $\CP\vdash P_i=Q_i$, hence
$\CP\vdash P_1\lef P_2\rig P_3=Q_1\lef Q_2\rig Q_3$. Now apply the auxiliary result.
\end{proof}

Recall that the symbol $\equiv$ denotes `syntactic equivalence'.

\begin{lemma}
\label{la:2.6}
For all basic forms $P$ and $Q$, $se(P)=se(Q)$ implies $P\equiv Q$.
\end{lemma}

\begin{proof} By structural induction on $P$. 
The base cases $P\in\{\tr,\fa\}$ are trivial.
If $P\equiv P_1\lef a\rig P_2$, then $Q\not\in\{\tr,\fa\}$
and $Q\not\equiv Q_1\lef b\rig Q_2$ if $b\ne a$, so $Q\equiv Q_1\lef a\rig Q_2$
and $se(P_i)=se(Q_i)$. By induction we find $P_i\equiv Q_i$, and hence $P\equiv Q$.
\end{proof}

\begin{definition}
\label{def:freevc}
\textbf{Free valuation congruence}, notation $=_\fr$, is defined on \PS\ as follows: 
\[P=_\fr Q\quad\iff\quad se(P)=se(Q).\]
\end{definition}

\begin{lemma}
\label{la:x}
Free valuation congruence is a congruence relation.
\end{lemma}

\begin{proof}
Let $P,Q,R\in\PS$ and assume $P=_\fr P'$, thus $se(P)=se(P')$. 
Then $se(P\lef Q\rig R)=se(Q)[\tr\mapsto se(P),\fa\mapsto se(R)]=
se(Q)[\tr\mapsto se(P'),\fa\mapsto se(R)]=se(P'\lef Q\rig R)$, and thus
$P\lef Q\rig R=_\fr P'\lef Q\rig R$. The two remaining cases
can be proved in a similar way.
\end{proof}

\begin{theorem}[Completeness of $\CP$ for closed terms]
\label{thm:A1}
For all $P,Q\in\PS$, 
\[\CP\vdash P=Q\quad\iff\quad P=_\fr Q.\]
\end{theorem}
\begin{proof}
We first prove $\Rightarrow$.\footnote{%
  Without loss of generality it can be assumed that substitutions happen first in equational proofs, 
  that is, the rule (Substitution) in Table~\ref{tab:el}
  may only be used when $s=t$ is an axiom in $\SCLe$ (see, e.g., \cite{Aceto}).}
By Lemma~\ref{la:x}, 
$=_\fr$ is a congruence relation and it easily follows
that closed instances of \CP-axioms are valid. In the case of axiom~\eqref{cp4} 
this follows from
\begin{align*}
se(P\lef(&Q\lef R\rig S)\rig U)\\
&=se(Q\lef R\rig S)[\tr\mapsto se(P),\fa\mapsto se(U)]\\
&=\big(se(R)[\tr\mapsto se(Q),\fa\mapsto se(S)]\big)\;[\tr\mapsto se(P),\fa\mapsto se(U)]\\
&=se(R)[\tr\mapsto se(Q)[\tr\mapsto se(P),\fa\mapsto se(U)],
\fa\mapsto \,se(S)[\tr\mapsto se(P),\fa\mapsto se(U)]]\\
&=se(R)[\tr\mapsto se(P\lef Q\rig U),\fa\mapsto se(P\lef S\rig U)]\\
&=se((P\lef Q\rig U)\lef R\rig(P\lef S\rig U)).
\end{align*}

In order to prove $\Leftarrow$, let $P=_\fr Q$. According to Lemma~\ref{la:2.5}
there exist basic forms
$P'$ and $Q'$ such that $\CP\vdash P=P'$ and $\CP\vdash Q=Q'$, so
$\CP\vdash P'=Q'$. By ($\Rightarrow$) we find $P'=_\fr Q'$, so
by Lemma~\ref{la:2.6}, $P'\equiv Q'$. Hence, $\CP\vdash P=P'= Q'=Q$.
\end{proof}

\end{document}